\documentclass[11pt]{amsart}
\usepackage[latin1]{inputenc}
\usepackage{amsmath}
\usepackage{amsfonts}
\usepackage{amssymb}
\usepackage{graphicx}
\usepackage{fourier}
\usepackage{dsfont}
\usepackage{enumerate}
\newtheorem{theorem}{Theorem}[section]
\newtheorem{lemma}[theorem]{Lemma}

\theoremstyle{definition}
\newtheorem{definition}[theorem]{Definition}
\newtheorem{conjecture}[theorem]{Conjecture}

\theoremstyle{remark}
\newtheorem{remark}[theorem]{Remark}

\numberwithin{equation}{section}

\author{Amit Einav}
\title{A Counter Example to Cercignani's Conjecture for the $d$ Dimensional Kac Model}
\thanks{The author was supported by ERC Grant MATKIT }
\address{Amit Einav \\ 
Department of Pure Mathematics and Mathematical Statistics \\ University of Cambridge}
\email{A.Einav@dpmms.cam.ac.uk}

\begin{document}

\maketitle

\begin{abstract}
Kac's $d$ dimensional model gives a linear, many particle, binary collision model from which, under suitable conditions, the celebrated Boltzmann equation, in its spatially homogeneous form, arise as a mean field limit. The ergodicity of the evolution equation leads to questions about the relaxation rate, in hope that such a rate would pass on the Boltzmann equation as the number of particles goes to infinity. This program, starting with Kac and his one dimensional 'Spectral Gap Conjecture' at 1956, finally reached its conclusion in a series of papers by authors such as Janvresse, Maslen, Carlen, Carvalho, Loss and Geronimo, but the hope to get a a limiting relaxation rate for the Boltzmann equation with this linear method was already shown to be unrealistic. A less linear approach, via a many particle version of Cercignani's conjecture, is the grounds for this paper. In our paper, we extend recent results by the author from the one dimensional Kac model to the $d$ dimensional one, showing that the entropy-entropy production ratio, $\Gamma_N$, still yields a very strong dependency in the number of particles of the problem when we consider the general case.
\end{abstract}

\section{Introduction}\label{sec: introduction}
One of the most important equations in the field of non equilibrium Statistical Physics is the celebrated Boltzmann equation. In its spatially homogeneous form it is given by:
\begin{equation}\label{eq: boltzmann equation}
\frac{\partial f}{\partial t}(v,t)=Q(f,f)(v,t),
\end{equation}
where $v\in\mathbb{R}^d$, $d\geq 2$ and
\begin{equation}\label{eq: boltzmann collision operator}
\begin{gathered}
Q(f,f)=\int_{\mathbb{R}^d \times \mathbb{S}^{d-1}}B\left(|v-v_\ast|,\cos(\vartheta) \right)\left(f\left(v^\prime \right)f\left(v_\ast^\prime \right)-
f\left(v \right)f\left(v_\ast\right)\right)dv_\ast d\omega , \\
v^\prime =\frac{v+v_\ast}{2}+\frac{\left\lvert v-v_\ast \right\rvert}{2}\cdot\omega , \\
v_\ast^\prime =\frac{v+v_\ast}{2}-\frac{\left\lvert v-v_\ast \right\rvert}{2}\cdot\omega.
\end{gathered}
\end{equation}
$v,v^\prime$ stand for the pre collision velocities and $\vartheta\in[0,\pi]$ is the deviation angle between $v-v_\ast$ and $v^\prime - v_\ast ^\prime$. The function $B$ is the Boltzmann collision kernel, affected by the physics of the problem, such as the cross section.\\
While physically motivated, to this day a proof of the derivation of (\ref{eq: boltzmann equation}) from the reversible Newtonian laws is missing in full. The main, and remarkable, progress in that area was done in 1973, by Lanford (see \cite{Lanford}), who managed to show the result for short times (shorter than the average time before we see collisions).\\ 
In his 1956 paper \cite{Kac}, Marc Kac introduced probability into the mix, and along with a new concept - 'Boltzmann Property' (what we now call chaotic families) - he managed derive a caricature of the spatially homogeneous Boltzmann equation in one dimensions as a mean field limit of his stochastic process. Kac considered a linear $N$-particle binary collision model with an evolution equation (the 'master equation') given by
\begin{equation}\label{eq: master equation in 1d}
\frac{\partial F_N}{\partial t}\left(v_1,\dots,v_N \right)=-N(I-Q)F_N\left(v_1,\dots,v_N \right),
\end{equation}
where 
\[QF\left(v_1,\dots,v_N \right)=\frac{1}{2\pi}\cdot \frac{2}{N(N-1)}\sum_{i<j}\int_{0}^{2\pi}F\left(v_1,\dots,v_i(\vartheta),\dots,v_j(\vartheta),\dots,v_N\right)d\vartheta ,\]
with 
\begin{equation}\label{eq: v_i(theta),v_j(theta)}
\begin{gathered}
v_i(\vartheta)=v_i\cos(\vartheta)+v_j\sin(\vartheta), \\
v_j(\vartheta)=-v_i\sin(\vartheta)+v_j\cos(\vartheta).
\end{gathered}
\end{equation}
Under the assumption of chaoticity, i.e. that the $k$-th marginal of $F_N$ converges to the $k$-tensorization of the limit of the first marginal, $f$ (where the limits are considered in the weak sense), Kac showed that $f$ satsfies the following spatially homogeneous 'Boltzmann equation':
\begin{equation}\label{eq: caricature boltzmann equation}
\frac{\partial f}{\partial t}(v,t)=\frac{1}{2\pi}\int_{\mathbb{R}}\int_0 ^{2\pi}\left(f\left(v(\vartheta)\right)f\left(v_\ast(\vartheta) \right)-
f\left(v \right)f\left(v_\ast\right)\right)dv_\ast d\vartheta ,
\end{equation}
where $v(\vartheta),v_\ast(\vartheta)$ are defined as in (\ref{eq: v_i(theta),v_j(theta)}). Note that a simple comparison of (\ref{eq: boltzmann equation}) with (\ref{eq: caricature boltzmann equation}) shows that in his model, Kac assumed that $B=1$, which is the less physical but very interesting mathematically case of the so called 'Grad Maxwell Molecules'. The reason behind this is the immense difficulty in mixing a collision function that depends on the relative velocities along with the jump process  (see \cite{Kac,MM}).\\
While the model itself wasn't completely physical, as it doesn't conserve momentum, it still gave rise to many interesting observations and results. The first one is    that the property of chaoticity \textit{propagates} with the evolution. This means that if we started with a chaotic family, then at each time $t$, the solution to (\ref{eq: master equation in 1d}) is still a chaotic family. The proof is a beautiful combinatorial argument along with an explicit expression to the solution (wild sums). Another important observation was that the evolution equation (\ref{eq: master equation in 1d}) is ergodic on $\mathbb{S}^{N-1}(\sqrt{N})$, implying that $\lim_{t\rightarrow\infty}F\left(t,v_1,\dots,v_n \right)=1$ for any fixed $N$. This led Kac to hope that a rate of relaxation of his linear equation can be bounded independently of $N$ and serve to prove a rate of relaxation to the associated Boltzmann equation. Denoting by 
\[\Delta_N= \inf_{F_N\in L^2_{sym}\left(\sqrt{N} \right)}\left\lbrace \frac{\left\langle F_N,N(I-Q)F_N \right\rangle}{\left\lVert F_N \right\rVert^2_{L^2\left(\sqrt{N},d\sigma^N \right)}}, \qquad F_N\perp 1 \right\rbrace ,\]
where $L^2_{sym}\left(\sqrt{N} \right)$ is the set of symmetric $L^2\left( \mathbb{S}^{N-1}(\sqrt{N}),d\sigma^N \right)$ functions and $d\sigma^N$ is the uniform probability measure on the sphere, Kac conjectured that $\liminf_{N\rightarrow\infty}\Delta_N >0$. This would lead to the following estimation:
\begin{equation}\label{eq: spectral gap estimation}
\left\lVert F_N(t) - 1 \right\rVert_{L^2_{sym}\left(\sqrt{N} \right)} \leq e^{-(\liminf_{N\rightarrow\infty}\Delta_N)t} \left\lVert F_N(0) - 1 \right\rVert_{L^2_{sym}\left(\sqrt{N} \right)}.
\end{equation}
The 'spectral gap' problem was investigated by many people including Janvresse (\cite{Jan}) and Maslen (\cite{Maslen}), and was finally given an explicit answer by Carlen, Carvahlo and Loss (\cite{CCL}) who managed to show that 
\[\Delta_N=\frac{N+2}{2(N-1)}.\]
Inequality (\ref{eq: spectral gap estimation}) along with the propagation of chaos would seemingly lead to an exponential decay to equilibrium of the first marginal, now that we know that Kac's conjecture is true, but a closer look shows this to be false. Indeed, intuitively speaking, being a chaotic family means that in some sense $F_N \sim f^{\otimes N}$. This leads to a very strong dependency of $N$ in the right term of (\ref{eq: spectral gap estimation}). One can find a chaotic family on the sphere, $F_N$, such that 
\[\left\lVert F_N \right\rVert_{L^2\left(\mathbb{S}^{N-1}(\sqrt{N}),d\sigma^N \right)}\geq C^N ,\]
where $C>1$, which leads to a relaxation time of order $N$.\\ 
The reason for the above catastrophe is the choice of $L^2$ as a reference norm along with the chaoticity requirement. A better norm-like function is required, one that is more amiable towards the chaoticity property.\\
Bearing that in mind, a natural quantity to investigate is the entropy. On the Kac sphere it is defined as 
\[H_N(F)=\int_{\mathbb{S}^{N-1}(\sqrt{N})}F \log F d\sigma^N .\]
The superiority of the entropy over the $L^2$ norm is given by its \textit{extensiviy} property: intuitively speaking, for chaotic families that satisfy $F_N \sim f^{\otimes N}$ we have that 
\[H_N(F_N)\approx N H(f \vert \gamma),\]
where $H(f \vert \gamma)=\int_{\mathbb{R}}f \log \left(f/ \gamma \right)$ and $\gamma$ is the standard Gaussian.\\
A related 'spectral gap' problem appeared: Noticing that 
\[D(F_N)=-\frac{\partial H_N(F_N)}{\partial t}=\left\langle \log F, N(I-Q)F \right\rangle\]
whenever $F_N$ is the solution to (\ref{eq: master equation in 1d}), one can ask if there exists $C>0$ such that 
\begin{equation}\label{eq: entropic spectral gap 1d}
\Gamma_N=\inf_{F\in L^2_{sym}\left(\mathbb{S}^{N-1}(\sqrt{N})\right)}\frac{D(F_N)}{H_N(F)}
\end{equation}
satisfies $\Gamma_N>C$? If it is true then a known inequility by Csisz$\'a$r, Kullback, Leibler and Pinsker shows that
\[\left\lVert F_N(t)d\sigma^N-d\sigma^N \right\rVert_{TV} \leq 2H_N(F_N(t))
\leq 2e^{-Ct}H_N(F_N(0)),\]
giving us a way to measure relaxation time of the marginals.\\ 
The above question is a variant of Cercignani's conjecture (see \cite{Cer}) known as the many particles Cercignani's conjecture. \\
The answer to that conjecture is \textit{No}. In his 2003 paper, \cite{Villani}, Villani managed to prove that $\Gamma_N \geq 2/(N-1)$ and conjectured that 
\begin{conjecture}\label{con: villani}
\begin{equation}\label{con: villani}
\Gamma_N = O\left( \frac{1}{N} \right).
\end{equation}
\end{conjecture}
In 2011, the author managed to show that for any $0<\eta<1$ there exists $C_\eta >0$ such that $\Gamma_N \leq C_\eta / N^\eta$ (see \cite{Einav}), giving a proof to an 'almost-$\epsilon$' version of Villani's conjecture and showing that in its full generality, the entropy-entropy production method doesn't give a much better result than the spectral gap approach.\\
While the one dimensional model itself posed, and still posses, many interesting problem, the fact that it is not very physical is a small deterrent. In his 1967 paper,  \cite{McKean}, McKean generalized Kac's model to a more realistic, momentum and energy conserving, $d$ dimensional model from which the real Boltzmann equation, (\ref{eq: boltzmann equation}), arose. McKean also extended the allowed collision kernels (though he still demanded that there won't be dependency on the relative velocity and that there would be no angular singularities) and showed propagation of chaos in a similar method to that of Kac. \\
The evolution equation to the simplest $d$-dimensional model, where $B=1$ (Grad Maxwellian Molecules), is given by 
\begin{equation}\label{eq: master equation in dd}
\frac{\partial F_N}{\partial t}\left(v_1,\dots,v_N \right)=-N(I-Q)F_N\left(v_1,\dots,v_N \right),
\end{equation}
where $v_1,\dots,v_N\in \mathbb{R}^d$ and 
\begin{equation}\label{eq: collision operator}
\begin{gathered}
QF\left(v_1,\dots,v_N \right)=\frac{2}{N(N-1)}\sum_{i<j} \\
\int_{\mathbb{S}^{d-1}}F\left(v_1,\dots,v_i(\omega),\dots,v_j(\omega),\dots,v_N\right)d\sigma^d ,
\end{gathered}
\end{equation}
with 
\begin{equation}\label{eq: v_i(omega) v_j(omega)}
\begin{gathered}
v_i(\omega)=\frac{v_i+v_j}{2}+\frac{\left\lvert v_i-v_j \right\rvert}{2}\cdot\omega , \\
v_j(\omega)=\frac{v_i+v_j}{2}-\frac{\left\lvert v_i-v_j \right\rvert}{2}\cdot\omega .
\end{gathered}
\end{equation}
The appropriate space is no longer the energy sphere $\mathbb{S}^{N-1}(\sqrt{N})$, but the Boltzmann sphere, defined by:
\begin{definition}\label{def: boltzmann sphere}
\begin{equation}\label{eq: boltzmann sphere}
\mathcal{S}_{B}^N(E,z)=\left\lbrace v_1,\dots,v_N\in \mathbb{R}^{d} \thickspace \Bigg\vert \thickspace \sum_{i=1}^N |v_i|^2=E \thickspace, \thickspace \sum_{i=1}^N v_i=z \right\rbrace.
\end{equation}
\end{definition}
with $E=N$ and $z=0$ for simplicity. For more information we refer the reader to \cite{CGL}.\\
The related spectral gap problem was solved in 2008 by Carlen, Geronimo and Loss (see \cite{CGL}), but a similar reasoning to that presented in the one dimensional case leads us to conclude that the spectral gap method is not suited to deal with chaotic families. \\
Like before, we define the entropy on the Boltzmann sphere as:
\begin{definition}\label{def: entropy}
\begin{equation}\label{eq: entropy}
H_N(F)=\int_{\mathcal{S}_B^N(N,0)}F \log F d\sigma^N_{N,0},
\end{equation}
where $d\sigma^N_{E,z}$ is the uniform probability measure on the Boltzmann sphere.
\end{definition}
One can ask now, similar to the one dimensional discussion, if a many particles Cercignani's conjecture holds in this case, or do we find the same situation as that of Conjecture \ref{con: villani}? \\
Defining:
\begin{definition}\label{def: entopic spectral gap dd}
\begin{equation}\label{eq: entropic spectral gap dd}
\Gamma_N=\inf \frac{D(F_N)}{H_N(F)},
\end{equation}
where $D(F_N)=\left\langle \log F, N(I-Q)F \right\rangle$ and the infimum is being taken over all symmetric probability densities over the Boltzmann sphere.
\end{definition}
we have that the main theorem of our paper is:
\begin{theorem}\label{thm: main theorem}
For any $0<\eta<1$ there exists a constant $C_\eta$, depending only on $\eta$, such that $\Gamma_N$, defined in (\ref{eq: entropic spectral gap dd}), satisfies
\begin{equation}\label{eq: Villani conjecture dd}
\Gamma_N \leq \frac{C_\eta}{N^{\eta}}.
\end{equation}
\end{theorem} 
The idea behind this proof is one that keeps repeating (see \cite{Einav,CCRLV}). An intuitive way to create a chaotic family on the Boltzmann sphere is by tensorising a one variable function (what we call our 'generating function'):
\[F_N\left(v_1,\dots,v_N \right)=\frac{\prod_{i=1}^N f(v_i)}{\mathcal{Z}_N\left(f,\sqrt{N},0 \right)}\]
where the normalization function $\mathcal{Z}_N$ is defined by
\[\mathcal{Z}_N\left(f,\sqrt{u},z \right)=\int_{\mathcal{S}^N_B(u,z)}\prod_{i=1}^N f(v_i) d\sigma^N_{u,z}\]
The new method, presented originally in our previous work on the one dimensional case (see  \cite{Einav}), that we use here is to allow the function $f$ to depend on $N$, and still control the normalization function in an explicit way. The additional dimensions and geometry of the problem cause technical difficulties than in the one dimensional case, manifesting mainly in the normalization function and an approximation theorem for it. More details on the difficulties and how we solved them are presented in Sections \ref{sec: preliminaries} and \ref{sec: normalization function}.\\
The above introduction is, by far, a mere glimpse into the Kac model and its relation to the Boltzmann equation. There are many more details and some remarkable proofs involved with this subject and we refer the reader to \cite{CCL,CCRLV,CGL,MM,VReview} to read more about it.\\
The paper is structured as follows: Section \ref{sec: preliminaries} will discuss some preliminaries, giving more information about the Boltzmann sphere and the normalization function. Section \ref{sec: normalization function} will contain our specific choice of 'generating function' and the approximation theorem of its normalization function, leading to Section \ref{sec: the main result} where we prove the main theorem. Section \ref{sec: final remarks} concludes with final words and some remarks and is followed by the Appendix, containing additional computation we found unnecessary to include in the main body of the paper.\\
\textbf{Acknowledgement:}
The author would like to thank Cl$\'e$ment Mouhot for many fruitful discussions and constant encouragement, as well as Kleber Carrapatoso for allowing him to read the preprint of his paper (\cite{Carr}), helping to bridge the dimension gap. 

\section{Preliminaries}\label{sec: preliminaries}
In this section we'll discuss a few preliminary results, mainly about the Boltzmann sphere and the normalization function $Z_N(f,\sqrt{u},z)$. Many of the results presented here can be found in \cite{Carr}, but we choose to present a variant of them for completion.
\subsection{The Boltzmann Sphere}\label{subsec: boltzmann sphere}
Recall Definition \ref{def: boltzmann sphere}, where the Boltzmann sphere was defined as
\[\mathcal{S}_{B}^N(E,z)=\left\lbrace v_1,\dots,v_N\in \mathbb{R}^{d} \thickspace \Bigg\vert \thickspace \sum_{i=1}^N |v_i|^2=E \thickspace, \thickspace \sum_{i=1}^N v_i=z \right\rbrace.\]
The term 'Boltzmann sphere' is evident from the following 'transformation':
\begin{equation}\label{eq: boltzmann sphere transformation}
U=RV,
\end{equation}
where $V=\left(v_1,\dots,v_N \right)^T$ and $R$ is the orthogonal matrix with rows given by
\[r_j=\frac{1}{\sqrt{j(j+1)}}\left(\sum_{i=1}^{j}e_i-je_{j+1}\right) \qquad 1\leq j \leq N-1,\]
\[r_N=\frac{\sum_{i=1}^N e_i}{\sqrt{N}},\]
where $e_j\in \mathbb{R}^N$ is the standard basis. Under (\ref{eq: boltzmann sphere transformation}) we see that 
\begin{equation}\label{eq: boltzmann sphere after transformation}
\mathcal{S}_B^N(E,z)=\left\lbrace u_1,\dots,u_N\in\mathbb{R}^d \thickspace \Bigg\vert \thickspace \sum_{i=1}^{N-1} |u_i|^2 = E-\frac{|z|^2}{N} \thickspace , \thickspace u_N=\frac{z}{\sqrt{N}} \right\rbrace ,
\end{equation}
giving us a sphere in a hyperplane of $d(N-1)$ dimensions of $\mathbb{R}^{dN}$ with radius $\sqrt{E-\frac{|z|^2}{N}}$.\\
Since we'll be interested in integration with respect to the uniform probability measure on the Boltzmann sphere, $d\sigma^N_{E,z}$, we will need the following Fubini-type formula:
\begin{theorem}\label{thm: fubini on the boltzmann sphere}
\begin{equation}
\begin{gathered}
\int_{\mathcal{S}_B^N(E,z)}Fd\sigma^N_{E,z}=\frac{\left\lvert \mathbb{S}^{d(N-j-1)-1} \right\rvert}{\left\lvert \mathbb{S}^{d(N-1)-1} \right\rvert}\cdot \frac{N^{\frac{d}{2}}}{(N-j)^\frac{d}{2} \left(E-\frac{|z|^2}{N}\right)^{\frac{d(N-1)-2}{2}}} \\
\int_{\Pi_j(E,z)}dv_1\dots dv_j\left(E-\sum_{i=1}^j |v_i|^2-\frac{\left\lvert z-\sum_{i=1}^j v_i \right\rvert^2}{N-j} \right)^{\frac{d(N-j-1)-2}{2}}\\
\int_{\mathcal{S}_B^{N-j}\left(E-\sum_{i=i}^j|v_i|^2,z-\sum_{i=1}^j v_i\right)}F d\sigma^{N-j}_{E-\sum_{i=i}^j|v_i|^2,z-\sum_{i=1}^j v_i} ,
\end{gathered}
\end{equation}
where $\Pi_j(E,z)=\left\lbrace \sum_{i=1}^j |v_i|^2 + \frac{|z-\sum_{i=1}^j v_i|^2}{N-j} \leq E \right\rbrace$.
\end{theorem}
We leave the proof to the Appendix (See Theorem \ref{thm: fubini on the boltzmann sphere app}).
\subsection{The Normalization Function}\label{subsec: normalization function preliminaries}
A key part of the proof of our main theorem lies with an approximation of the appropriate normalization function. While the true approximation theorem will be discussed in Section \ref{sec: normalization function}, we present here some basic probabilistic interpretation of it as a prelude to the proof.\\
As was mentioned before, the normalization function for a suitable function $f$ is defined as:
\begin{definition}\label{def: normalization function}
\begin{equation}\label{eq: normalization function def}
\mathcal{Z}_N\left(f,\sqrt{r},z \right)=\int_{\mathcal{S}_B^N(r,z) }f^{\otimes N}d\sigma^N_{r,z}.
\end{equation}
\end{definition}
\begin{lemma}\label{lem: probabilistic interpretation of the normalization function}
Let $V$ be a random variable with values in $\mathbb{R}^d$ and law $f$. Let $h$ be the law of the couple $\left(V,|V|^2 \right)$ then 
\begin{equation}\label{eq: probabilistic interpretation of the normalization function}
\mathcal{Z}_N \left(f,\sqrt{u},z\right)=\frac{2N^{\frac{d}{2}}h^{\ast N}(z,u)}{\left\lvert \mathbb{S}^{d(N-1)-1} \right\rvert \left(u-\frac{|z|^2}{N}\right)^\frac{d(N-1)-2}{2}}.
\end{equation}
\end{lemma}
\begin{proof}
Let $\varphi\in C_b$ be a function of $\sum_{i=1}^N v_i$ and $\sum_{i=1}^N |v_i|^2$. By the definition
\[\mathds{E}\varphi=\int_{\mathbb{R}^{dN}}\varphi \left(\sum_{i=1}^N v_i, \sum_{i=1}^N |v_i|^2 \right)f^{\otimes N}\left(v_1,\dots,v_N \right)dv_1 \dots dv_N .\]
Using (\ref{eq: boltzmann sphere transformation}) we can rewrite the above as
\[\int_{\mathbb{R}^{dN}}\varphi \left(\sqrt{N}u_N, \sum_{i=1}^N |u_i|^2 \right)f^{\otimes N}\circ R^{-1}\left(u_1,\dots,u_N \right)du_1 \dots du_N\]
\[=\int_{\mathbb{R}^d}du_N \int_{\mathbb{R}^{d(N-1)}}\varphi \left(\sqrt{N}u_N, \sum_{i=1}^N |u_i|^2 \right)f^{\otimes N}\circ R^{-1}\left(u_1,\dots,u_N \right)du_1 \dots du_N\]
\[=\int_{\mathbb{R}^d}du_N \int_0^\infty dr\cdot r^{d(N-1)-1}\varphi \left(\sqrt{N}u_N, r^2+|u_N|^2 \right)\int_{\mathbb{S}^{d(N-1)-1}}f^{\otimes N}\circ R^{-1}ds^{d(N-1)}\]
\[=\int_{\mathbb{R}^d}du_N \int_0^\infty dr\left\lvert \mathbb{S}^{d(N-1)-1} \right\rvert r^{d(N-1)-1}\varphi \left(\sqrt{N}u_N, r^2+|u_N|^2 \right)\int_{\mathbb{S}^{d(N-1)-1}}f^{\otimes N}\circ R^{-1}d\gamma^{d(N-1)} ,\]
where $d\gamma$ is the uniform probability measure on the sphere.\\
At this point we notice that 
\[\int_{\mathbb{S}^{d(N-1)-1}}f^{\otimes N}\circ R^{-1}d\gamma^{d(N-1)}=\int_{\sum_{i=1}^{N-1} |u_i|^2=r^2, u_N}f^{\otimes N}\circ R^{-1}d\gamma^{d(N-1)}\]
\[=\int_{\sum_{i=1}^N |v_i|^2=r^2+|u_N|^2, \sum_{i=1}^N v_i= \sqrt{N}u_N}f^{\otimes N}d\sigma^N=\mathcal{Z}_N\left(f,\sqrt{r^2+|u_N|^2},\sqrt{N}u_N \right) .\]
Using the change of variables $z=\sqrt{N}u_N$ and $w=r^2+|u_N|^2$ yields
\[\mathds{E}\varphi=\int_{\mathbb{R}^d}du_N \int_0^\infty dr\left\lvert \mathbb{S}^{d(N-1)-1} \right\rvert \left(r^2\right)^\frac{d(N-1)-1}{2}\varphi \left(\sqrt{N}u_N, r^2+|u_N|^2 \right)\]
\[\cdot\mathcal{Z}_N\left(f,\sqrt{r^2+|u_N|^2},\sqrt{N}u_N \right)=\int_{\mathbb{R}^d}du_N \int_0^\infty dw\frac{\left\lvert \mathbb{S}^{d(N-1)-1} \right\rvert}{2} \left(w-|u_N|^2\right)^\frac{d(N-1)-2}{2}\varphi \left(\sqrt{N}u_N, w \right)\]
\[\cdot\mathcal{Z}_N\left(f,\sqrt{w},\sqrt{N}u_N \right)=\int_{\mathbb{R}^d}dz \int_0^\infty dw\frac{\left\lvert \mathbb{S}^{d(N-1)-1} \right\rvert}{2N^{\frac{d}{2}}} \left(w-\frac{|z|^2}{N}\right)^\frac{d(N-1)-2}{2}\varphi \left(z, w \right)\cdot\mathcal{Z}_N\left(f,\sqrt{w},z \right).\]
On the other hand, denoting by $s_N$ the law of the couple 
\[\sum_{i=1}^N \left(V_i, \left\lvert V_i \right\rvert^2 \right) = \left(\sum_{i=1}^N V_i, \sum_{i=1}^n \left\lvert V_i \right\rvert^2 \right) ,\] we find that 
\[\mathds{E}\varphi=\int_{\mathbb{R}^d \times [0,\infty)}\varphi\left(z,w \right)s_N(z,w)dzdw.\]
This leads to the conclusion that 
\[\mathcal{Z}_N\left(f,\sqrt{w},z \right)=\frac{2N^{\frac{d}{2}}s_N(z,w)}{\left\lvert \mathbb{S}^{d(N-1)-1} \right\rvert \left(w-\frac{|z|^2}{N}\right)^\frac{d(N-1)-2}{2}},\]
and the result follows using a known theorem in Probability Theory.
\end{proof}
The fact that convolution itself gives us a function and not just a law is discussed in \cite{Carr}. In our particular case we'll prove that we indeed get a well defined function upon a very specific choice of law $f$.\\
We conclude this section with the connection between the law of $V$ and the couple $\left(V,|V|^2 \right)$.
\begin{lemma}\label{lem: law of h}
Let $f$ be a density function for the random variable $V$. Then, the law of the couple $\left(V,|V|^2 \right)$, denoted by $h$, is given by 
\[dh(v,u)=f(v)\delta_{u=|v|^2}(u)dvdu.\]
\end{lemma}
\begin{proof}
let $\varphi\in C_b$ be a function of $v$. Then 
\[\mathds{E}\varphi=\int_{\mathbb{R}^d}\varphi(v)f(v)dv.\]
On the other hand 
\[\mathds{E}\varphi=\int_{\mathbb{R}^d \times [0,\infty)}\varphi(v)h(v,u)dvdu.\]
Since every function of the couple $(v,|v|^2)$ is actually a function of $v$. The result follows.
\end{proof}

\section{The Normalization Function and its Approximation}\label{sec: normalization function}
The core of the proof of the main theorem of our paper lies in understanding how the normalization function of a particular \textit{changing} family of densities behaves asymptotically on the Kac sphere, following ideas presented in \cite{Einav}. \\
The first step we must take is to define the 'generating function'. This is a very natural choice following the trends of \cite{BC,CCRLV,Einav}. \\

\begin{definition}
We denote by
\begin{equation}\label{eq: definition of the generating function}
f_{\delta}(v)=\delta M_{\frac{1}{2d\delta}}(v)+(1-\delta) M_{\frac{1}{2d(1-\delta)}}(v),
\end{equation}
 where $M_a(v)=\frac{1}{\left(2\pi a \right)^{\frac{d}{2}}}e^{-\frac{|v|^2}{2a}}$.
\end{definition}

The main theorem of this section is the following:
\begin{theorem}\label{thm: main approximation theorem for the normalization function}
Let $f_{\delta_N}$ be as in (\ref{eq: definition of the generating function}) where $\delta_N=\frac{1}{N^{1-\eta}}$ with $\frac{2\beta}{1+2\beta}<\eta<\frac{(3+d)\beta}{1+3\beta+\frac{d}{2}+d\beta}$ and $0<\beta<1$ arbitrary. Then 
\begin{equation}
\begin{gathered}
\sup_{u\in[0,\infty),v\in\mathbb{R}^d}\left\lvert \frac{\left\lvert \mathbb{S}^{d(N-1)-1}\right\rvert \left(u-\frac{|v|^2}{N} \right)^{\frac{d(N-1)-2}{2}}}{2N^{\frac{d}{2}}}\mathcal{Z}_N\left(f_{\delta_N},\sqrt{u},v \right)-\gamma_N(u,v) \right\rvert \\
\leq \frac{\epsilon(N)}{\Sigma_{\delta_N} N^{\frac{d+1}{2}}},
\end{gathered}
\end{equation}
where $\gamma_N(u,v)=\frac{d^{\frac{d}{2}}}{\Sigma_{\delta_N} N^{\frac{d+1}{2}}}\cdot \frac{e^{-\frac{d|v|^2}{2N}}\cdot e^{-\frac{(u-N)^2}{2\Sigma_{\delta_N}^2 N}}}{(2\pi)^{\frac{d+1}{2}}}$, $\Sigma_{\delta_N}^2=\frac{d+2}{4d\delta_N (1-\delta_N)}-1$ and  $\lim_{N\rightarrow\infty}\epsilon(N)=0$.
\end{theorem}

\begin{remark}\label{rem: approximation theorem follows recent ones}
The above approximation theorem gives a similar result to the one presented in \cite{Carr}, however a closer inspection of our choice of 'generating function' shows a difference in the definition of $\Sigma$. We believe this difference manifests itself due to the dependency of $\delta$ in $N$, appearing as a different dimension factor. 
\end{remark}
The proof of the above theorem is quite technical and will occupy us for the rest of this section. We encourage the reader to skip the rest of this section at first reading, and jump to Section \ref{sec: the main result} to see how the approximation theorem serves to prove the main result.\\
Before we begin we'd like to state a few technical Lemmas.
\begin{lemma}\label{lem: the fourier transform of the generating function}
Let $f_{\delta}$ be as defined in (\ref{eq: definition of the generating function}). Then 
\begin{equation}\label{eq: the fourier transform of h associated to the generating function}
\widehat{h_\delta}(p,t)=\frac{\delta e^{-\frac{\pi^2 |p|^2}{d\delta+2\pi i t}}}{\left(1+\frac{2\pi i t}{d\delta}\right)^{\frac{d}{2}}}+\frac{(1-\delta) e^{-\frac{\pi^2 |p|^2}{d(1-\delta)+2\pi i t}}}{\left(1+\frac{2\pi i t}{d(1-\delta)}\right)^{\frac{d}{2}}},
\end{equation}
where $h_\delta$ is associated to $f_\delta$ via Lemma \ref{lem: law of h} and the Fourier transform is defined in the measure sense.
\end{lemma}
\begin{proof}
We begin with the known Fourier transform of the Gaussian
\[\int_{\mathbb{R}}e^{-\beta x^2}e^{-2\pi i x \xi}dx=\sqrt{\frac{\pi}{\beta}}e^{-\frac{\pi^2 \xi^2}{\beta}}\]
for $\beta>0$. Since both sides are clearly analytic in $\beta$ for $\textbf{Re}\beta>0$ we find that the equality is still true in that domain.\\
Denoting by $h_a$ the law associated to the couple $\left(V,|V|^2 \right)$ where $V$ has law $M_a$, we notice that by the above remark, Lemma \ref{lem: law of h} and the definition of the Fourier transform of a measure:
\[\widehat{h_a}(p,t)=\int_{\mathbb{R}^d \times [0,\infty)}e^{-2\pi i (p\circ v +tu)}dh(v,u)=\int_{\mathbb{R}^d}M_a(v)e^{-2\pi i (p\circ v +t|v|^2)}\]
\[=\prod_{i=1}^d \frac{1}{\sqrt{2\pi a}}\int_{\mathbb{R}}e^{-v_i^2 \left(\frac{1}{2a}+2\pi i t \right)}\cdot e^{-2\pi i p_i v_i}dv_i=\prod_{i=1}^d \frac{e^{-\frac{\pi^2 p_i^2}{\frac{1}{2a}+2\pi i t}}}{\sqrt{1+4\pi i a t}}=\frac{e^{-\frac{2a\pi^2 |p|^2}{1+4\pi i a t}}}{\left(1+4\pi i a t\right)^{\frac{d}{2}}}.\]
Thus the result follows immediately from the definition of $f_\delta$ and the linearity of the Fourier transform.
\end{proof}
At this point we'll explain why the convolution in (\ref{eq: probabilistic interpretation of the normalization function}) yields a function. The proof of the following Lemma is provided in the Appendix.
\begin{lemma}\label{lem: the convolution yields a function}
Let $h_\delta$ be associated to $f_\delta$ via Lemma \ref{lem: law of h} where $f_\delta$ is defined in (\ref{eq: definition of the generating function}). Then $\widehat{h_\delta}^n\in L^{q}\left(\mathbb{R}^d \times [0,\infty) \right)$ when $n>\frac{2(1+d)}{qd}$. In particular, for every $n>\frac{2(1+d)}{d}$ we have that $\widehat{h_\delta}^n\in L^2\left(\mathbb{R}^d \times [0,\infty) \right)\cap L^1 \left(\mathbb{R}^d \times [0,\infty) \right)$ and thus $h^{\ast n}$ can be viewed as a density.
\end{lemma}
Next, we state and prove a couple of integral estimations.
\begin{lemma}\label{lem: one dimensional gaussian estimation}
For any $\alpha,\beta>0$ we have that 
\begin{align}\label{eq: one dimensional gaussian estimation}
\int_{x>\beta} e^{-\alpha x^2}dx \leq \sqrt{\frac{\pi}{4\alpha}}e^{-\frac{\alpha \beta^2}{2}}. \\
\int_{x>\beta} xe^{-\alpha x^2}dx \leq \frac{e^{-\frac{\alpha \beta^2}{2}}}{2\alpha} .\\
\int_{0<x\leq\beta} e^{-\alpha x^2}dx\leq \sqrt{\frac{\pi}{4\alpha}}\sqrt{1-e^{-2\alpha \beta^2}}.
\end{align}
\end{lemma}
\begin{proof}
This follows immediately from the next estimations
\[ \int_{x>\beta} e^{-\alpha x^2}dx=\frac{1}{\sqrt{\alpha}}\int_{x>\sqrt{\alpha}\beta}e^{-x^2} dx
=\frac{1}{\sqrt{\alpha}}\sqrt{\int_{x,y>\sqrt{\alpha}\beta}e^{-x^2-y^2} dxdy} \]
\[\leq \frac{1}{\sqrt{\alpha}} \sqrt{\int_{x^2+y^2>\alpha\beta^2, x>0,y>0 }e^{-x^2-y^2} dxdy}=\sqrt{\frac{\pi}{2\alpha}}\sqrt{\int_{r>\sqrt{\alpha}\beta}re^{-r^2}dr} =\sqrt{\frac{\pi}{4\alpha}}e^{-\frac{\alpha \beta^2}{2}}.\]
\[\int_{x>\beta} xe^{-\alpha x^2}dx=\frac{1}{\alpha} \int_{x>\sqrt{\alpha}\beta}xe^{-x^2} dx = \frac{e^{-\alpha \beta^2}}{2\alpha} \leq \frac{e^{-\frac{\alpha \beta^2}{2}}}{2\alpha}. \]
\[\int_{x\leq\beta} e^{-\alpha x^2}dx\leq \frac{1}{\sqrt{\alpha}}\sqrt{\int_{x^2+y^2\leq 2\alpha\beta^2,x>0,y>0}e^{-x^2-y^2} dxdy}=\sqrt{\frac{\pi}{2 \alpha}}\sqrt{\int_{r<\sqrt{2\alpha}\beta}re^{-r^2}dr}\]
\[=\sqrt{\frac{\pi}{4\alpha}}\sqrt{1-e^{-2\alpha \beta^2}}.\]
\end{proof}
\begin{lemma}\label{lem: Gaussian estimation}
\begin{equation}\label{eq: Gaussian estimation}
\int_{|x|>\beta} |x|^{m}e^{-\alpha |x|^2}d^d x \leq \frac{C_{m,d} \max \left(\beta^{m+d-2},\beta^{m+d-4},\dots,1 \right)}{\min \left(\alpha, \alpha^2, \dots, \alpha^{\left[\frac{m+d+2}{2}\right]} \right)}e^{-\frac{\alpha \beta^2}{2}},  
\end{equation}
where $C_{m.d}$ is a constant depending only on $m$ and $d$.
\end{lemma}
\begin{proof}
First we notice that 
\[\int_{|x|>\beta} |x|^{m}e^{-\alpha |x|^2}d^d x=C_{d,m} \int_{r>\beta} r^{m+d-1} e^{-\alpha r^2}dr, \]
where $C_{d,m}$ is a constant depending only on $m$ and $d$.\\
Lemma \ref{lem: one dimensional gaussian estimation} tells us that 
\[\int_{x>\beta} x^j e^{-\alpha x^2}dx \leq \frac{Ce^{-\frac{\alpha \beta^2}{2}}}{\min(\alpha,\sqrt{\alpha})},\]
where $j=0,1$. For $j>1$ we have that 
\[\int_{x>\beta} x^j e^{-\alpha x^2}dx=\frac{1}{\alpha^{\frac{j+1}{2}}}\int_{x>\sqrt{\alpha}\beta}x^j e^{-x^2} dx\]
\[=\frac{1}{\alpha^{\frac{j+1}{2}}}\left(-\frac{x^{j-1}e^{-x^2}}{2}\vert_{\sqrt{\alpha}\beta}^{\infty}+\frac{j-1}{2}\int_{x>\sqrt{\alpha}\beta}x^{j-2} e^{-x^2} dx \right)\]
\[=\frac{1}{\alpha^{\frac{j+1}{2}}}\left(\frac{\alpha^{\frac{j-1}{2}}\beta^{j-1}e^{-\alpha \beta^2}}{2}+\frac{j-1}{2}\int_{x>\sqrt{\alpha}\beta}x^{j-2} e^{-x^2} dx \right).\]
Continuing to integrate by parts yields 
\[\int_{x>\beta} x^j e^{-\alpha x^2}dx\leq \frac{C_j}{\alpha^{\frac{j+1}{2}}}
\left( \alpha^{\frac{j-1}{2}}\beta^{j-1}e^{-\alpha \beta^2}+\alpha^{\frac{j-3}{2}}\beta^{j-3}e^{-\alpha \beta^2}+\dots+ \int_{x>\sqrt{\alpha}\beta} x^{\tilde{j}} e^{- x^2}dx \right)\]
\[=C_j \left( \frac{\beta^{j-1}e^{-\alpha \beta^2}}{\alpha}+\frac{\beta^{j-3}e^{-\alpha \beta^2}}{\alpha^2}+\dots+ \frac{1}{\alpha^{\frac{j+1}{2}}}\int_{x>\sqrt{\alpha}\beta} x^{\tilde{j}} e^{- x^2}dx  \right),\]
where $C_j$ is a constant depending only on $j$ and $\tilde{j}=0,1$. Using our previous estimation we conclude that 
\[\int_{x>\beta} x^j e^{-\alpha x^2}dx \leq \frac{C_j \max \left(\beta^{j-1},\beta^{j-3},\dots,1 \right)}{\min \left(\alpha, \alpha^2, \dots, \alpha^{\left[\frac{j+3}{2}\right]} \right)}e^{-\frac{\alpha \beta^2}{2}},\]
completing the proof.
\end{proof}
\begin{remark}\label{rem: special case alpha>1 and beta<1}
In the special case where $\alpha\geq 1$ and $\beta\leq 1$ we get the estimation
\begin{equation}\label{eq: Special Gaussian estimation with alpha>1 and beta<1}
\int_{|x|>\beta} |x|^{m}e^{-\alpha |x|^2}d^d x \leq
\frac{C_{m,d}}{\alpha}e^{-\frac{\alpha \beta^2}{2}}.
\end{equation}
\end{remark}
Lastly, we notice three things:
\begin{enumerate}
\item It is easy to show that $\widehat{\gamma_N}(p,t)=\widehat{\gamma_1}^N(p,t)$ where 
\[\gamma_1(p,t)=e^{-\frac{2\pi^2 |p|^2}{d}}\cdot e^{-2\pi i t}\cdot e^{-2\pi^2 \Sigma^2_{\delta}t^2}.\]
\item An estimation we'll constantly use is the following: For any $0\leq k \leq N-1$ we have that 
\begin{equation}\label{eq: estimation of the kth power of abs h time the N-k th power of gamma 1}
\begin{gathered}
\left\lvert \widehat{h}(p,t) \right\rvert ^k \left\lvert \widehat{\gamma_1}(p,t) \right\rvert ^{N-k-1}\\
\leq \sum_{j=0}^{k} \left( \begin{array} {c} k \\ j \end{array} \right)
\frac{\delta^j}{\left(1+\frac{4\pi^2 t^2}{d^2\delta^2} \right)^{\frac{dj}{4}}}
\frac{(1-\delta)^{k-j}}{\left(1+\frac{4\pi^2 t^2}{d^2(1-\delta)^2} \right)^{\frac{d(k-j)}{4}}}\\
\cdot e^{-\pi^2 |p|^2 \left(\frac{jd\delta}{d^2 \delta^2+4\pi^2 t^2}+ \frac{(k-j)d(1-\delta)}{d^2 (1-\delta)^2+4\pi^2 t^2}+\frac{2(N-k-1)}{d} \right)}
e^{-2\pi^2(N-k-1) \Sigma^2 t^2}.
\end{gathered}
\end{equation}
\item $\left\lvert \widehat{h}(p,t) \right\rvert \leq 1$ and $\left\lvert \widehat{\gamma_1}(p,t) \right\rvert \leq 1$.
\end{enumerate}
In order to prove an our approximation theorem we need to divide the phasespace domain $\mathbb{R}^d\times\mathbb{R}$ into three domains. The following subsections deal with that division, and end in the proof of Theorem \ref{thm: main approximation theorem for the normalization function}.
\subsection{Large $t$, any $p$: $\vert t \vert > \frac{d \delta^{1+\beta}}{4\pi}$. }\label{subsec: large t any p}
The main theorem of this subsection is the following:
\begin{theorem}\label{thm: large t case}
\begin{equation}\label{eq: large t case}
\begin{gathered}
\iint_{\mathbb{R}^d \times |t|>\frac{d\delta^{1+\beta}}{4\pi}}\left\lvert \widehat{h}^N(p,t)-\widehat{\gamma_1}^N(p,t) \right\rvert dpdt \\
 \leq \frac{N C_d}{\left( N-2 \right)^\frac{d+1}{2} \Sigma}\cdot e^{-\frac{d^2 (N-2)\Sigma^2 \delta^{2+2\beta}}{32}} +\frac{NC_d}{\Sigma}\left(1-\frac{d\delta^{1+2\beta}}{16}+\delta^{1+4\beta}\xi(\delta) \right)^{\frac{N}{2}}\cdot e^{-\frac{d^2\Sigma^2\delta^{2+2\beta}}{32}} \\
 +C_d \left(1-\frac{d\delta^{1+2\beta}}{16}+\delta^{1+4\beta}\xi(\delta) \right)^{N-5},
\end{gathered}
\end{equation}
where $C_d$ is a constant depending only on $d$ and $\xi$ is analytic in $|x|<\frac{1}{2}$.
\end{theorem}
In order to prove the above theorem we need a series of Lemmas and small computations.\\
We start by noticing that due to (\ref{eq: estimation of the kth power of abs h time the N-k th power of gamma 1}) we have
\begin{equation}\label{eq: estimation of the p-integral of the kth power of abs h time the N-k th power of gamma 1}
\begin{gathered}
\int_{\mathbb{R}^d} \left\lvert \widehat{h}(p,t) \right\rvert ^k \left\lvert \widehat{\gamma_1}(p,t) \right\rvert ^{N-k-1} dp\\
\leq C_d \sum_{j=0}^{k} \left( \begin{array} {c} k \\ j \end{array} \right)
\frac{\delta^j}{\left(1+\frac{4\pi^2 t^2}{d^2\delta^2} \right)^{\frac{dj}{4}}}
\frac{(1-\delta)^{k-j}}{\left(1+\frac{4\pi^2 t^2}{d^2(1-\delta)^2} \right)^{\frac{d(k-j)}{4}}} \\
 \cdot \frac{e^{-2\pi^2(N-k-1) \Sigma^2 t^2}}{\left( \frac{jd\delta}{d^2 \delta^2+4\pi^2 t^2}+ \frac{(k-j)d(1-\delta)}{d^2 (1-\delta)^2+4\pi^2 t^2}+\frac{2(N-k-1)}{d} \right)^{\frac{d}{2}}},
\end{gathered}
\end{equation}
where $C_d=\frac{1}{\pi^d}\int_{\mathbb{R}^d} e^{-|z|^2}dz$. \\
Next, we see that 
\begin{equation}\label{eq: main estimation of the p integration}
\begin{gathered}
\frac{1}{\left( \frac{jd\delta}{d^2 \delta^2+4\pi^2 t^2}+\frac{(k-j)d(1-\delta)}{d^2 (1-\delta)^2+4\pi^2 t^2}+\frac{2(N-k-1)}{d} \right)^{\frac{d}{2}}} \\
 \leq \min \left(\frac{\left( d^2 \delta^2+4\pi^2 t^2 \right)^{\frac{d}{2}}}{\left( jd\delta \right)^{\frac{d}{2}}},\frac{\left( d^2 (1-\delta)^2+4\pi^2 t^2 \right)^{\frac{d}{2}}}{\left( (k-j)d(1-\delta) \right)^{\frac{d}{2}}},\frac{d^{\frac{d}{2}}}{\left( 2(N-k-1) \right)^{\frac{d}{2}}}\right).
\end{gathered}
\end{equation}
Also, since 
\[d^2 \delta^2+4\pi^2 t^2 \leq d^2+4 \pi^2 t^2 ,\]
\[d^2 (1-\delta)^2+4\pi^2 t^2 \leq d^2+4\pi^2 t^2 ,\]
when $0<\delta<1$, we have that
\begin{equation}\label{eq: mid estimation}
\max \left( \left(d^2 \delta^2+4\pi^2 t^2\right)^{\frac{d}{2}}, \left(d^2 (1-\delta)^2+4\pi^2 t^2 \right)^{\frac{d}{2}} \right)
\leq A_d (1+|t|^d),
\end{equation}
where $A_d$ is a constant depending only on $d$.
We are now ready to state and prove our first Lemma.
\begin{lemma}\label{lem: large t sum up to N/2}
\begin{equation}\label{eq: large t sum up to N/2}
\begin{gathered}
\sum_{k=0}^{\left[ \frac{N}{2} \right]} \iint_{\mathbb{R}^d \times |t|>\frac{d\delta^{1+\beta}}{4\pi}}  \left\lvert \widehat{h}(p,t)-\widehat{\gamma_1}(p,t) \right\rvert \left\lvert \widehat{h}(p,t)\right\rvert ^k \left\lvert \widehat{\gamma_1}(p,t) \right\rvert ^{N-k-1} dpdt \\
\leq \frac{N C_d}{\left( N-2 \right)^\frac{d+1}{2} \Sigma}\cdot e^{-\frac{d^2 (N-2)\Sigma^2 \delta^{2+2\beta}}{32}},
\end{gathered}
\end{equation}
where $C_d$ is a constant depending only on $d$.
\end{lemma}

\begin{proof}
Since $\left\lvert \widehat{h}(p,t) \right\rvert \leq 1$ and $\left\lvert \widehat{\gamma_1}(p,t) \right\rvert \leq 1$, we find that along with inequality (\ref{eq: estimation of the p-integral of the kth power of abs h time the N-k th power of gamma 1}), inequality (\ref{eq: main estimation of the p integration}) and the fact that $k\leq \frac{N}{2}$ we have
\[\sum_{k=0}^{\left[ \frac{N}{2} \right]} \iint_{\mathbb{R}^d \times |t|>\frac{d\delta^{1+\beta}}{4\pi}}  \left\lvert \widehat{h}(p,t)-\widehat{\gamma_1}(p,t) \right\rvert \left\lvert \widehat{h}(p,t)\right\rvert ^k \left\lvert \widehat{\gamma_1}(p,t) \right\rvert ^{N-k-1} dpdt\]
\[\leq \frac{2C_d d^\frac{d}{2}}{(N-2)^\frac{d}{2}} \sum_{k=0}^{\left[ \frac{N}{2} \right]} \int_{|t|>\frac{d \delta^{1+\beta}}{4\pi}} \sum_{j=0}^{k} \left( \begin{array} {c} k \\ j \end{array} \right) \frac{\delta^j}{\left(1+\frac{4\pi^2 t^2}{d^2\delta^2} \right)^{\frac{dj}{4}}} \frac{(1-\delta)^{k-j}}{\left(1+\frac{4\pi^2 t^2}{d^2(1-\delta)^2} \right)^{\frac{d(k-j)}{4}}}  \cdot e^{-\pi^2(N-2) \Sigma^2 t^2}\]
\[=\frac{2C_d d^\frac{d}{2}}{(N-2)^\frac{d}{2}} \sum_{k=0}^{\left[ \frac{N}{2} \right]} \int_{|t|>\frac{d \delta^{1+\beta}}{4\pi}} \left(\frac{\delta}{\left(1+\frac{4\pi^2 t^2}{d^2\delta^2} \right)^{\frac{d}{4}}} + \frac{(1-\delta)}{\left(1+\frac{4\pi^2 t^2}{d^2(1-\delta)^2} \right)^{\frac{d}{4}}} \right)^k    \cdot e^{-\pi^2(N-2) \Sigma^2 t^2} \]
\[\leq \frac{2C_d d^\frac{d}{2}}{(N-2)^\frac{d}{2}} \sum_{k=0}^{\left[ \frac{N}{2} \right]} \int_{|t|>\frac{d \delta^{1+\beta}}{4\pi}} e^{-\pi^2(N-2) \Sigma^2 t^2}
\leq \frac{NC_d d^\frac{d}{2}}{(N-2)^\frac{d}{2}} \int_{|t|>\frac{d \delta^{1+\beta}}{4\pi}} e^{-\pi^2(N-2) \Sigma^2 t^2}\]
\[ \leq \frac{N C_d}{\left( N-2 \right)^\frac{d+1}{2} \Sigma}\cdot e^{-\frac{d^2 (N-2)\Sigma^2 \delta^{2+2\beta}}{32}},\]
where the last inequality is due to Lemma \ref{lem: one dimensional gaussian estimation}. This concludes the proof.
\end{proof}
For the following Lemmas we will need the next observation: In our domain we have that
\[\frac{\delta}{\left(1+\frac{4\pi^2 t^2}{d^2 \delta^2} \right)^{\frac{d}{4}}}+\frac{1-\delta}{\left(1+\frac{4\pi^2 t^2}{d^2 (1-\delta)^2} \right)^{\frac{d}{4}}}\leq \frac{\delta}{\left(1+\frac{\delta^{2\beta}}{4} \right)^{\frac{d}{4}}}+\frac{1-\delta}{\left(1+\frac{\delta^{2+2\beta}}{4(1-\delta)^2} \right)^{\frac{d}{4}}}\]
\[=\delta \left(1-\frac{d\delta^{2\beta}}{16}+\delta^{4\beta}\phi\left(\delta^{2\beta}\right) \right)+(1-\delta) \left(1-\frac{d\delta^{2+2\beta}}{16(1-\delta)^2}+\frac{\delta^{4+4\beta}}{(1-\delta)^4}\phi\left(\frac{\delta^{2+2\beta}}{(1-\delta)^2}\right) \right),\]
where $\phi$ in analytic in $|x|<\frac{1}{2}$. Opening the parenthesis leads to
\[=1-\frac{d\delta^{1+2\beta}}{16}-\frac{d\delta^{2+2\beta}}{16(1-\delta)}+\delta^{1+4\beta}\phi\left(\delta^{2\beta}\right) 
+\frac{\delta^{4+4\beta}}{(1-\delta)^3}\phi\left(\frac{\delta^{2+2\beta}}{(1-\delta)^2} \right).\]
which we can write as the inequality:
\begin{equation}\label{eq: smallish t exponential term of delta}
\frac{\delta}{\left(1+\frac{4\pi^2 t^2}{d^2 \delta^2} \right)^{\frac{d}{4}}}+\frac{1-\delta}{\left(1+\frac{4\pi^2 t^2}{d^2 (1-\delta)^2} \right)^{\frac{d}{4}}}
\leq 1-\frac{d\delta^{1+2\beta}}{16}+\delta^{1+4\beta}\xi(\delta).
\end{equation}
where $\xi$ is analytic in $|x|<\frac{1}{2}$.\\

We're now ready to state and prove our second Lemma.
\begin{lemma}\label{lem: large t sum form N/2 up to N-2}
\begin{equation}\label{eq: large t sum from N/2 up to N-2}
\begin{gathered}
\sum_{k={\left[ \frac{N}{2} \right]+1}}^{N-2} \iint_{\mathbb{R}^d \times |t|>\frac{d\delta^{1+\beta}}{4\pi}}  \left\lvert \widehat{h}(p,t)-\widehat{\gamma_1}(p,t) \right\rvert \left\lvert \widehat{h}(p,t)\right\rvert ^k \left\lvert \widehat{\gamma_1}(p,t) \right\rvert ^{N-k-1} dpdt \\
\leq \frac{NC_d}{2\Sigma}\left(1-\frac{d\delta^{1+2\beta}}{16}+\delta^{1+4\beta}\xi(\delta) \right)^{\frac{N}{2}}\cdot e^{-\frac{d^2\Sigma^2\delta^{2+2\beta}}{32}},
\end{gathered}
\end{equation}
where $C_d$ is a constant depending only on $d$ and $\xi$ is analytic in $|x|<\frac{1}{2}$.
\end{lemma}

\begin{proof}
Like in the proof of Lemma \ref{lem: large t sum up to N/2} we'll be using inequalities (\ref{eq: estimation of the p-integral of the kth power of abs h time the N-k th power of gamma 1}), (\ref{eq: main estimation of the p integration}), inequality (\ref{eq: smallish t exponential term of delta}) and the fact that $N-k-1\geq 1$ to conclude that 
\[\sum_{k={\left[ \frac{N}{2} \right]+1}}^{N-2} \iint_{\mathbb{R}^d \times |t|>\frac{d\delta^{1+\beta}}{4\pi}}  \left\lvert \widehat{h}(p,t)-\widehat{\gamma_1}(p,t) \right\rvert \left\lvert \widehat{h}(p,t)\right\rvert ^k \left\lvert \widehat{\gamma_1}(p,t) \right\rvert ^{N-k-1} dpdt\]
\[\leq\frac{2C_d d^\frac{d}{2}}{2^\frac{d}{2}} \sum_{k=\left[ \frac{N}{2} \right]+1}^{N-2} \int_{|t|>\frac{d \delta^{1+\beta}}{4\pi}} \sum_{j=0}^{k} \left( \begin{array} {c} k \\ j \end{array} \right) \frac{\delta^j}{\left(1+\frac{4\pi^2 t^2}{d^2\delta^2} \right)^{\frac{dj}{4}}} \frac{(1-\delta)^{k-j}}{\left(1+\frac{4\pi^2 t^2}{d^2(1-\delta)^2} \right)^{\frac{d(k-j)}{4}}}  \cdot e^{-\pi^2\Sigma^2 t^2} \]
\[=C_d \sum_{k=\left[ \frac{N}{2} \right]+1}^{N-2} \int_{|t|>\frac{d \delta^{1+\beta}}{4\pi}} \left(\frac{\delta^{1+\beta}}{\left(1+\frac{4\pi^2 t^2}{d^2\delta^2} \right)^{\frac{d}{4}}} + \frac{(1-\delta)}{\left(1+\frac{4\pi^2 t^2}{d^2(1-\delta)^2} \right)^{\frac{d}{4}}} \right)^k \cdot e^{-\pi^2\Sigma^2 t^2} \]
\[\leq \frac{NC_d}{2}\left(1-\frac{d\delta^{1+2\beta}}{16}+\delta^{1+4\beta}\xi(\delta)\right)^{\frac{N}{2}}\int_{|t|>\frac{d \delta^{1+\beta}}{4\pi}} e^{-\pi^2\Sigma^2 t^2},\]
and Lemma \ref{lem: one dimensional gaussian estimation} yields the final estimation.\end{proof}

Lastly, we have the following Lemma:
\begin{lemma}\label{lem: t large k=N-1}
\begin{equation}\label{eq: large t k=N-1}
\begin{gathered}
\iint_{\mathbb{R}^d \times |t|>\frac{d\delta^{1+\beta}}{4\pi}}  \left\lvert \widehat{h}(p,t)-\widehat{\gamma_1}(p,t) \right\rvert \left\lvert \widehat{h}(p,t)\right\rvert ^{N-1} dpdt \\
\leq C_d \left(1-\frac{d\delta^{1+2\beta}}{16}+\delta^{1+4\beta}\xi(\delta) \right)^{N-5},
\end{gathered}
\end{equation}
where $C_d$ is a constant depending only on $d$ and $\xi$ is analytic in $|x|<\frac{1}{2}$.
\end{lemma}
\begin{proof}
Using inequality (\ref{eq: estimation of the p-integral of the kth power of abs h time the N-k th power of gamma 1}), (\ref{eq: main estimation of the p integration}) and (\ref{eq: mid estimation}) with $k=N-1$ we find that 
\[ \iint_{\mathbb{R}^d \times |t|>\frac{d\delta^{1+\beta}}{4\pi}}  \left\lvert \widehat{h}(p,t)-\widehat{\gamma_1}(p,t) \right\rvert \left\lvert \widehat{h}(p,t)\right\rvert ^{N-1} dpdt\]
\[ \leq C_d \int_{|t|>\frac{d\delta^{1+\beta}}{4\pi}}\sum_{j=0}^{N-1} \left( \begin{array} {c} N-1 \\ j \end{array} \right)
\frac{\delta^j}{\left(1+\frac{4\pi^2 t^2}{d^2\delta^2} \right)^{\frac{dj}{4}}}
\frac{(1-\delta)^{N-1-j}}{\left(1+\frac{4\pi^2 t^2}{d^2(1-\delta)^2} \right)^{\frac{d(N-1-j)}{4}}} \]
\[ \cdot \min \left( \frac{\left(1+|t|^d \right)}{\left(jd\delta \right)^{\frac{d}{2}}},\frac{\left(1+|t|^d \right)}{\left((N-1-j)d(1-\delta) \right)^{\frac{d}{2}}}\right)dt .\]
For $\delta<\frac{1}{2}$ and $0\leq j \leq N-1$ we find that 
\[\min \left( \frac{\left(1+|t|^d \right)}{\left(jd\delta \right)^{\frac{d}{2}}},\frac{\left(1+|t|^d \right)}{\left((N-1-j)d(1-\delta) \right)^{\frac{d}{2}}}\right)
\leq \frac{1+|t|^d}{\left(\frac{\delta(N-1)}{2}\right)^{\frac{d}{2}}}.\]
Thus, our desired expression is bounded above by
\[ \frac{C_d}{(\delta(N-1))^\frac{d}{2}}\int_{|t|>\frac{d\delta^{1+\beta}}{4\pi}}\left(\frac{\delta}{\left(1+\frac{4\pi^2 t^2}{d^2\delta^2} \right)^{\frac{d}{4}}}+
\frac{(1-\delta)}{\left(1+\frac{4\pi^2 t^2}{d^2(1-\delta)^2} \right)^{\frac{d}{4}}} \right) ^{N-1}\left(1+|t|^d \right)dt\]
\[\leq \frac{C_d}{(\delta(N-1))^\frac{d}{2}}\left(1-\frac{d\delta^{1+2\beta}}{16}+\delta^{1+4\beta}\xi(\delta) \right)^{N-5}\]
\[\int_{|t|>\frac{d\delta^{1+\beta}}{4\pi}}\left(\frac{\delta}{\left(1+\frac{4\pi^2 t^2}{d^2\delta^2} \right)^{\frac{d}{4}}}+
\frac{(1-\delta)}{\left(1+\frac{4\pi^2 t^2}{d^2(1-\delta)^2} \right)^{\frac{d}{4}}} \right) ^{4}\left(1+|t|^d \right)dt. \]
Once we'll show that 
\[\int_{|t|>\frac{d\delta^{1+\beta}}{4\pi}}\left(\frac{\delta}{\left(1+\frac{4\pi^2 t^2}{d^2\delta^2} \right)^{\frac{d}{4}}}+
\frac{(1-\delta)}{\left(1+\frac{4\pi^2 t^2}{d^2(1-\delta)^2} \right)^{\frac{d}{4}}} \right) ^{4}\left(1+|t|^d \right)dt \leq C_d,\]
the proof will be done.\\
Indeed, 
\[\int_{|t|>\frac{d\delta^{1+\beta}}{4\pi}}\left(\frac{\delta}{\left(1+\frac{4\pi^2 t^2}{d^2\delta^2} \right)^{\frac{d}{4}}}+
\frac{(1-\delta)}{\left(1+\frac{4\pi^2 t^2}{d^2(1-\delta)^2} \right)^{\frac{d}{4}}} \right) ^{4}\left(1+|t|^d \right)dt\]
\[ \leq\int_{\mathbb{R}}\left(\frac{\delta}{\left(1+\frac{4\pi^2 t^2}{d^2\delta^2} \right)^{\frac{d}{4}}}+
\frac{(1-\delta)}{\left(1+\frac{4\pi^2 t^2}{d^2(1-\delta)^2} \right)^{\frac{d}{4}}} \right) ^{4}\left(1+|t|^d \right)dt\]
\[=\int_{|t|\leq 1}\left(\frac{\delta}{\left(1+\frac{4\pi^2 t^2}{d^2\delta^2} \right)^{\frac{d}{4}}}+
\frac{(1-\delta)}{\left(1+\frac{4\pi^2 t^2}{d^2(1-\delta)^2} \right)^{\frac{d}{4}}} \right) ^{4}\left(1+|t|^d \right)dt\] 
\[+ \int_{|t|> 1}\left(\frac{\delta}{\left(1+\frac{4\pi^2 t^2}{d^2\delta^2} \right)^{\frac{d}{4}}}+
\frac{(1-\delta)}{\left(1+\frac{4\pi^2 t^2}{d^2(1-\delta)^2} \right)^{\frac{d}{4}}} \right) ^{4}\left(1+|t|^d \right)dt\]
\[\leq \int_{|t|\leq 1}2dt+\int_{|t|> 1}\left(\frac{d^{\frac{d}{2}}\delta^{\frac{d+2}{2}}}{\left(2\pi t \right)^{\frac{d}{2}}}+
\frac{d^{\frac{d}{2}}(1-\delta)^{\frac{d+2}{2}}}{\left(2\pi t\right)^{\frac{d}{2}}}\right) ^{4}2|t|^d dt\]
\[\leq 4+\int_{|t|> 1}\left(\frac{d^{\frac{d}{2}}}{\left(2\pi t \right)^{\frac{d}{2}}}+
\frac{d^{\frac{d}{2}}}{\left(2\pi t\right)^{\frac{d}{2}}}\right) ^{4}2|t|^d dt = 4 + \frac{2^5 d^{2d}}{(2\pi)^{2d}}\int_{|t|>1}\frac{dt}{|t|^d} =C_d .\]
\end{proof}
\begin{proof}[Proof of Theorem \ref{thm: large t case}]
This follows from Lemma \ref{lem: large t sum up to N/2}, Lemma \ref{lem: large t sum form N/2 up to N-2}, Lemma \ref{lem: t large k=N-1} and the estimation
\[\left\lvert \widehat{h}^N(p,t)-\widehat{\gamma_1}^N(p,t) \right\rvert \leq \left\lvert \widehat{h}(p,t)-\widehat{\gamma_1}(p,t) \right\rvert \sum_{k=0}^{N-1}\left\lvert \widehat{h}(p,t)\right\rvert^k \left\lvert \widehat{\gamma_1}(p,t) \right\rvert^{N-k-1}.\]
\end{proof}

\subsection{small $t$, large $p$: $|t|\leq \frac{d\delta^{1+\beta}}{4\pi}$ and $|p|>\eta$ }\label{subsec: small t large p}
The main theorem of this subsection is:
\begin{theorem}\label{thm: small t large p case}
\begin{equation}\label{eq: small t large p case}
\begin{gathered}
\iint_{|p|>\eta \times |t|\leq\frac{d\delta^{1+\beta}}{4\pi}}\left\lvert \widehat{h}^N(p,t)-\widehat{\gamma_1}^N(p,t) \right\rvert dpdt\leq \frac{N\delta^{1+\beta} C_de^{-\frac{(N-2)\eta^2}{4d}}}{N-2},
\end{gathered}
\end{equation}
where $C_d$ is a constant depending only on $d$.
\end{theorem}
Again, some Lemmas and computations are needed before we can prove the above.\\
To begin with, we notice that we can't use (\ref{eq: estimation of the kth power of abs h time the N-k th power of gamma 1}) any more as the domain of the $p$ integration changed. Instead, we use the same pre-integration estimation along with Remark \ref{rem: special case alpha>1 and beta<1} to find that
\begin{equation}\label{eq: estimation of the p integration with large t}
\begin{gathered}
\int_{|p|>\eta} \left\lvert \widehat{h}(p,t) \right\rvert ^k \left\lvert \widehat{\gamma_1}(p,t) \right\rvert ^{N-k-1} dp \\
 \leq C_d \sum_{j=0}^{k} \left( \begin{array} {c} k \\ j \end{array} \right)
\frac{\delta^j}{\left(1+\frac{4\pi^2 t^2}{d^2\delta^2} \right)^{\frac{dj}{4}}}
\frac{(1-\delta)^{k-j}}{\left(1+\frac{4\pi^2 t^2}{d^2(1-\delta)^2} \right)^{\frac{d(k-j)}{4}}}\\
 \cdot \frac{e^{-2\pi^2(N-k-1) \Sigma^2 t^2}e^{-\frac{\pi^2\left(\frac{jd\delta}{d^2\delta^2+4\pi^2 t^2}+\frac{(k-j)d(1-\delta)}{d^2(1-\delta)^2+4\pi^2 t^2}+\frac{N-k-1}{d} \right)\eta^2}{2}}}{\left( \frac{jd\delta}{d^2 \delta^2+4\pi^2 t^2}+ \frac{(k-j)d(1-\delta)}{d^2 (1-\delta)^2+4\pi^2 t^2}+\frac{2(N-k-1)}{d} \right)}.
\end{gathered}
\end{equation}
We need to justify the usage of the mentioned remark: 
In our domain $|t|\leq \frac{d\delta^{1+\beta}}{4\pi}<\frac{d\delta}{4\pi}$, and so 
\[d^2\delta^2+4\pi^2 t^2 \leq \frac{5d^2\delta^2}{4}.\]
Similarly, since $\delta<1-\delta$ we have that 
\[d^2(1-\delta)^2+4\pi^2 t^2 \leq \frac{5d^2(1-\delta)^2}{4},\]
leading us to conclude that, with the notation of Lemma \ref{lem: Gaussian estimation}: 
\[\alpha=\pi^2\left(\frac{jd\delta}{d^2\delta^2+4\pi^2 t^2}+\frac{(k-j)d(1-\delta)}{d^2(1-\delta)^2+4\pi^2 t^2}+\frac{N-k-1}{d} \right)\]
\[ \geq \pi^2 \left(\frac{4j}{5d\delta}+\frac{4(k-j)}{5d(1-\delta)}+\frac{N-k-1}{d}\right). \]
If $j\geq 1$ then $\frac{4j}{5d\delta}\geq \frac{4}{5d\delta}>1$ when $\delta$ is small enough.\\
If $k\leq\frac{N}{2}$ then $\frac{N-k-1}{d}>\frac{N-2}{2d}>1$ for large enough $N$.\\
If $j=0$ and $k>\frac{N}{2}$ then $\frac{4(k-j)}{5d(1-\delta)}\geq \frac{2N}{5d(1-\delta)}>1$ again. \\
In any case, $\alpha>1$.\\
Also, $\beta=\frac{d\delta^{1+\beta}}{4\pi}<1$ for small enough $\delta$, and so we managed to justify (\ref{eq: estimation of the p integration with large t}).\\
We are now ready to state and prove our first Lemma.
\begin{lemma}\label{lem: small t large p sum up to N/2}
\begin{equation}\label{eq: small t large p sum up to N/2}
\begin{gathered}
\sum_{k=0}^{\left[ \frac{N}{2} \right]} \iint_{|p|>\eta \times |t|\leq\frac{d\delta^{1+\beta}}{4\pi}}  \left\lvert \widehat{h}(p,t)-\widehat{\gamma_1}(p,t) \right\rvert \left\lvert \widehat{h}(p,t)\right\rvert ^k \left\lvert \widehat{\gamma_1}(p,t) \right\rvert ^{N-k-1} dpdt \\
\leq \frac{N\delta^{1+\beta} C_de^{-\frac{(N-2)\eta^2}{4d}}}{N-2},
\end{gathered}
\end{equation}
where $C_d$ is a constant depending only on $d$.
\end{lemma}
\begin{proof}
Since for $k\leq \frac{N}{2}$
\[\frac{e^{-2\pi^2(N-k-1) \Sigma^2 t^2}e^{-\frac{\pi^2\left(\frac{jd\delta}{d^2\delta^2+4\pi^2 t^2}+\frac{(k-j)d(1-\delta)}{d^2(1-\delta)^2+4\pi^2 t^2}+\frac{N-k-1}{d} \right)\eta^2}{2}}}{\left( \frac{jd\delta}{d^2 \delta^2+4\pi^2 t^2}+ \frac{(k-j)d(1-\delta)}{d^2 (1-\delta)^2+4\pi^2 t^2}+\frac{2(N-k-1)}{d} \right)}\leq \frac{e^{-\frac{(N-2)\eta^2}{4d}}}{\left( \frac{(N-2)}{d} \right)},\]
we have that due to inequality (\ref{eq: estimation of the p integration with large t})
\[\sum_{k=0}^{\left[ \frac{N}{2} \right]} \iint_{|p|>\eta \times |t|\leq\frac{d\delta^{1+\beta}}{4\pi}}  \left\lvert \widehat{h}(p,t)-\widehat{\gamma_1}(p,t) \right\rvert \left\lvert \widehat{h}(p,t)\right\rvert ^k \left\lvert \widehat{\gamma_1}(p,t) \right\rvert ^{N-k-1} dpdt\]
\[ \leq \frac{2dC_de^{-\frac{(N-2)\eta^2}{4d}}}{N-2}\sum_{k=0}^{\left[ \frac{N}{2} \right]}\int_{|t|\leq \frac{d\delta^{1+\beta}}{4\pi}}\sum_{j=0}^{k} \left( \begin{array} {c} k \\ j \end{array} \right)
\frac{\delta^j}{\left(1+\frac{4\pi^2 t^2}{d^2\delta^2} \right)^{\frac{dj}{4}}}
\frac{(1-\delta)^{k-j}}{\left(1+\frac{4\pi^2 t^2}{d^2(1-\delta)^2} \right)^{\frac{d(k-j)}{4}}} dt\]
\[=\frac{2dC_de^{-\frac{(N-2)\eta^2}{4d}}}{N-2}\sum_{k=0}^{\left[ \frac{N}{2} \right]}\int_{|t|\leq \frac{d\delta^{1+\beta}}{4\pi}}\left(\frac{\delta}{\left(1+\frac{4\pi^2 t^2}{d^2 \delta^2} \right)^{\frac{d}{4}}}+\frac{1-\delta}{\left(1+\frac{4\pi^2 t^2}{d^2 (1-\delta)^2} \right)^{\frac{d}{4}}} \right)^kdt\]
\[\leq \frac{dNC_de^{-\frac{(N-2)\eta^2}{4d}}}{N-2}\cdot \frac{d\delta^{1+\beta}}{4\pi},\]
which concludes the proof.
\end{proof}
Next, we notice that 
\[\frac{e^{-\frac{\pi^2\left(\frac{jd\delta}{d^2\delta^2+4\pi^2 t^2}+\frac{(k-j)d(1-\delta)}{d^2(1-\delta)^2+4\pi^2 t^2}+\frac{N-k-1}{d} \right)\eta^2}{2}}}{\left( \frac{jd\delta}{d^2 \delta^2+4\pi^2 t^2}+ \frac{(k-j)d(1-\delta)}{d^2 (1-\delta)^2+4\pi^2 t^2}+\frac{2(N-k-1)}{d} \right)}\]
\[\leq \min \left(\frac{(d^2 \delta^2+4\pi^2 t^2)e^{-\frac{\pi^2 d \delta j \eta^2}{2(d^2 \delta^2+4\pi^2 t^2)}}}{jd\delta},\frac{(d^2 (1-\delta)^2+4\pi^2 t^2)e^{-\frac{\pi^2 d (1-\delta)(k-j) \eta^2}{2(d^2 (1-\delta)^2+4\pi^2 t^2)}}}{(k-j)d(1-\delta)} \right)
\]
\[ \leq \min \left(\frac{5 d \delta e^{-\frac{2\pi^2 j \eta^2}{5d \delta}}}{4j},\frac{5 d (1-\delta) e^{-\frac{2\pi^2 (k-j) \eta^2}{5d (1-\delta)}}}{4(k-j)} \right).\]
Thus
\begin{equation}\label{eq: delicate estimation of the |p|>eta integrated term}
\frac{e^{-\frac{\pi^2\left(\frac{jd\delta}{d^2\delta^2+4\pi^2 t^2}+\frac{(k-j)d(1-\delta)}{d^2(1-\delta)^2+4\pi^2 t^2}+\frac{N-k-1}{d} \right)\eta^2}{2}}}{\left( \frac{jd\delta}{d^2 \delta^2+4\pi^2 t^2}+ \frac{(k-j)d(1-\delta)}{d^2 (1-\delta)^2+4\pi^2 t^2}+\frac{2(N-k-1)}{d} \right)} \leq \frac{5d}{4}\cdot \min \left(\frac{e^{-\frac{2\pi^2 j \eta^2}{5d \delta}}}{j},\frac{e^{-\frac{2\pi^2 (k-j) \eta^2}{5d (1-\delta)}}}{(k-j)} \right).
\end{equation}
The second Lemma follows:
\begin{lemma}\label{lem: small t large p sum from N/2 to N-1}
\begin{equation}\label{eq: small t large p sum from N/2 to N-1}
\begin{gathered}
\sum_{k=\left[ \frac{N}{2} \right]}^{N-1} \iint_{|p|>\eta \times |t|\leq\frac{d\delta^{1+\beta}}{4\pi}}  \left\lvert \widehat{h}(p,t)-\widehat{\gamma_1}(p,t) \right\rvert \left\lvert \widehat{h}(p,t)\right\rvert ^k \left\lvert \widehat{\gamma_1}(p,t) \right\rvert ^{N-k-1} dpdt \\
\leq \frac{N C_d \delta^{1+\beta} e^{-\frac{\pi^2(N-2)\eta^2}{10d(1-\delta)}}}{N-2},
\end{gathered}
\end{equation}
where $C_d$ is a constant depending only on $d$.
\end{lemma}
\begin{proof}
Due to inequality (\ref{eq: estimation of the p integration with large t}) and (\ref{eq: delicate estimation of the |p|>eta integrated term}) we find that 
\[\sum_{k=\left[ \frac{N}{2} \right]}^{N-1} \iint_{|p|>\eta \times |t|\leq\frac{d\delta^{1+\beta}}{4\pi}}  \left\lvert \widehat{h}(p,t)-\widehat{\gamma_1}(p,t) \right\rvert \left\lvert \widehat{h}(p,t)\right\rvert ^k \left\lvert \widehat{\gamma_1}(p,t) \right\rvert ^{N-k-1} dpdt\]
\[\leq \frac{5 d C_d}{2} \sum_{k=\left[ \frac{N}{2} \right]}^{N-1}\int_{|t|\leq \frac{d\delta^{1+\beta}}{4\pi}}\sum_{j=0}^{k} \left( \begin{array} {c} k \\ j \end{array} \right)
\frac{\delta^j}{\left(1+\frac{4\pi^2 t^2}{d^2\delta^2} \right)^{\frac{dj}{4}}}
\frac{(1-\delta)^{k-j}}{\left(1+\frac{4\pi^2 t^2}{d^2(1-\delta)^2} \right)^{\frac{d(k-j)}{4}}} \]
\[ \cdot \min \left(\frac{e^{-\frac{2\pi^2 j \eta^2}{5d \delta}}}{j},\frac{e^{-\frac{2\pi^2 (k-j) \eta^2}{5d (1-\delta)}}}{(k-j)} \right)dt \]
\[\leq C_d \sum_{k=\left[ \frac{N}{2} \right]}^{N-1}\int_{|t|\leq \frac{d\delta^{1+\beta}}{4\pi}}\sum_{j=0}^{\left[\frac{k}{2} \right]} \left( \begin{array} {c} k \\ j \end{array} \right)
\frac{\delta^j}{\left(1+\frac{4\pi^2 t^2}{d^2\delta^2} \right)^{\frac{dj}{4}}}
\frac{(1-\delta)^{k-j}}{\left(1+\frac{4\pi^2 t^2}{d^2(1-\delta)^2} \right)^{\frac{d(k-j)}{4}}}  \cdot \frac{e^{-\frac{2\pi^2 (k-j) \eta^2}{5d (1-\delta)}}}{(k-j)}dt \]
\[+C_d \sum_{k=\left[ \frac{N}{2} \right]}^{N-1}\int_{|t|\leq \frac{d\delta^{1+\beta}}{4\pi}}\sum_{j=\left[\frac{k}{2} \right]+1}^k \left( \begin{array} {c} k \\ j \end{array} \right)
\frac{\delta^j}{\left(1+\frac{4\pi^2 t^2}{d^2\delta^2} \right)^{\frac{dj}{4}}}
\frac{(1-\delta)^{k-j}}{\left(1+\frac{4\pi^2 t^2}{d^2(1-\delta)^2} \right)^{\frac{d(k-j)}{4}}}  \cdot \frac{e^{-\frac{2\pi^2 j \eta^2}{5 d \delta}}}{j}dt\]
\[ \leq  C_d \sum_{k=\left[ \frac{N}{2} \right]}^{N-1}\int_{|t|\leq \frac{d\delta^{1+\beta}}{4\pi}}\frac{2e^{-\frac{\pi^2 k \eta^2}{5d(1-\delta)}}}{k}\sum_{j=0}^{\left[\frac{k}{2} \right]} \left( \begin{array} {c} k \\ j \end{array} \right)
\frac{\delta^j}{\left(1+\frac{4\pi^2 t^2}{d^2\delta^2} \right)^{\frac{dj}{4}}}
\frac{(1-\delta)^{k-j}}{\left(1+\frac{4\pi^2 t^2}{d^2(1-\delta)^2} \right)^{\frac{d(k-j)}{4}}}dt\]
\[+C_d \sum_{k=\left[ \frac{N}{2} \right]}^{N-1}\int_{|t|\leq \frac{d\delta^{1+\beta}}{4\pi}}\frac{2e^{-\frac{\pi^2 k \eta^2}{5d\delta}}}{k}\sum_{j=\left[\frac{k}{2} \right]+1}^k \left( \begin{array} {c} k \\ j \end{array} \right)
\frac{\delta^j}{\left(1+\frac{4\pi^2 t^2}{d^2\delta^2} \right)^{\frac{dj}{4}}}
\frac{(1-\delta)^{k-j}}{\left(1+\frac{4\pi^2 t^2}{d^2(1-\delta)^2} \right)^{\frac{d(k-j)}{4}}}dt\]
\[ \leq \frac{2C_d e^{-\frac{\pi^2(N-2)\eta^2}{10d(1-\delta)}}}{N-2}\sum_{k=\left[ \frac{N}{2} \right]}^{N-1}\int_{|t|\leq \frac{d\delta^{1+|beta}}{4\pi}}\sum_{j=0}^k \left( \begin{array} {c} k \\ j \end{array} \right)
\frac{\delta^j}{\left(1+\frac{4\pi^2 t^2}{d^2\delta^2} \right)^{\frac{dj}{4}}}
\frac{(1-\delta)^{k-j}}{\left(1+\frac{4\pi^2 t^2}{d^2(1-\delta)^2} \right)^{\frac{d(k-j)}{4}}}dt\]
\[= \frac{2C_d e^{-\frac{\pi^2(N-2)\eta^2}{10d(1-\delta)}}}{N-2}\sum_{k=\left[ \frac{N}{2} \right]}^{N-1}\int_{|t|\leq \frac{d\delta^{1+\beta}}{4\pi}}\left(\frac{\delta}{\left(1+\frac{4\pi^2 t^2}{d^2 \delta^2} \right)^{\frac{d}{4}}}+\frac{1-\delta}{\left(1+\frac{4\pi^2 t^2}{d^2 (1-\delta)^2} \right)^{\frac{d}{4}}} \right)^kdt,\]
from which the result follows.
\end{proof}
\begin{proof}[Proof of Theorem \ref{thm: small t large p case}]
This follows from Lemma \ref{lem: small t large p sum up to N/2}, Lemma \ref{lem: small t large p sum from N/2 to N-1}, the fact that $\frac{\pi^2}{10(1-\delta)}>\frac{1}{4}$ and the inequality mentioned at the proof of Theorem \ref{thm: large t case}.
\end{proof}

\subsection{Small $t$, small $p$: $|t|<\frac{d\delta^{1+\beta}}{4\pi}$ and $|p|\leq\eta=\delta^{\frac{1}{2}+\beta}$}\label{subsec: small t small p}
The main result of this subsection is
\begin{theorem}\label{thm: t tiny p small case}
\begin{equation}\label{eq: t tiny small p case}
\begin{gathered}
\iint_{|p|\leq\delta^{\frac{1}{2}+\beta} \times |t|\leq \frac{d\delta^{1+\beta}}{4\pi}}\left\lvert \widehat{h}^N(p,t)-\widehat{\gamma_1}^N(p,t) \right\rvert dpdt \\
\leq  \frac{C_d}{\Sigma}\delta^{\frac{3}{2}+4\beta+\frac{d}{2}+d\beta}+ \frac{C_d \sqrt{N}}{\Sigma}\delta^{1+3\beta+\frac{d}{2}+d\beta},
\end{gathered}
\end{equation}
where $C_d$ is a constant depending only on $d$.
\end{theorem}
We start by the simple observation that in this domain
\[\left\lvert \Sigma^2 t \right\rvert \leq \frac{(d+2)\delta^{\beta}}{16\pi(1-\delta)}
<\frac{(d+2)\delta^{\beta}}{8\pi}, \]
when $\delta<\frac{1}{2}$.\\
The main difficulty in our domain is the need to have a more precise approximation to the functions involved. We start with the easier amongst the two:
\begin{lemma}\label{lem: approximation of gamma_1}
\begin{equation}\label{eq: approximation of gamma_1}
\widehat{\gamma_1}(p,t)=\left(1-2\pi i t -2\pi^2 t^2 (\Sigma^2+1)+t^3 g(t) \right)\left(1-\frac{2\pi^2|p|^2}{d}+|p|^4 f\left(|p|^2 \right) \right),
\end{equation}
where $g,f$ are entire and there exist constants $M_0,M_1$, depending only on $d$, such that 
\[|g(t)|\leq M_0+\frac{M_1}{\delta}.\]
\[|f\left(|p|^2 \right)|\leq M_0.\]
\end{lemma}
\begin{proof}
Using the approximation $e^x=1+x+\frac{x^2}{2}+x^3 \phi(x)$, where $\phi$ is entire, we find that 
\[e^{-\frac{2\pi^2|p|^2}{d}}=1-\frac{2\pi^2|p|^2}{d}+\frac{4\pi^4|p|^4}{d^2}\phi_1\left(\frac{2\pi^2|p|^2}{d} \right),\]
\[e^{-2\pi i t}=1-2\pi i t -\frac{4\pi^2 t^2}{2}-8\pi^3 t^3 \phi(2\pi i t),\]
and
\[e^{-2\pi^2\Sigma^2 t^2}=1-2\pi^2\Sigma^2 t^2+\frac{4\pi^4\Sigma^4 t^4}{2}+8\pi^6\Sigma^6 t^6 \phi \left( 2\pi^2\Sigma^2 t^2 \right),\]
where $\phi_1$ is entire. Thus
\[e^{-2\pi i t} \cdot e^{-2\pi^2\Sigma^2 t^2}=1-2\pi i t -2\pi^2 t^2 (\Sigma^2+1) +\pi^3 t^3 \left(4i\Sigma^2 - 8\phi(2\pi i t)\right) \]
\[+4\pi^4\Sigma^2 t^4+16\pi^5 t^5\Sigma^2\phi(2\pi i t)+\pi^4 t^4 \left(2\Sigma^4+8\pi^2\Sigma^6t^2\phi\left(2\pi^2\Sigma^2t^2 \right) \right)e^{-2\pi i t}\]
\[=1-2\pi i t -2\pi^2 t^2 (\Sigma^2+1)+t^3 g(t).\]
We clearly have that $g(t)$ is entire, and
\[|g(t)|\leq 4\pi^3\Sigma^2+8\pi^3\left\lvert \phi(2\pi i t) \right\rvert +4\pi^4\Sigma^2 |t|+ 16\pi^5 t^2 \Sigma^2\left\lvert \phi(2\pi i t) \right\rvert+2\pi^4 |t|\Sigma^4\]
\[+8\pi^6|t|^3\Sigma^6\left\lvert \phi\left(2\pi^2\Sigma^2 t^2 \right) \right\rvert
\leq \frac{2\pi^3(d+2)}{d\delta}+8\pi^3 M_{sup}+\frac{\pi^3(d+2)\delta^{\beta}}{2} +\frac{\pi^3 d(d+2)\delta^{1+2\beta}M_{sup}}{2} \]
\[+\frac{\pi^4 (d+2)^2 \delta^\beta}{8d\delta}+\frac{\pi^3(d+2)^3\delta^{3\beta}M_{sup}}{64},\]
where $M_{sup}=\sup_{|x|<1}|\phi(x)|$. A simpler argument on $f$ leads to the desired result.
\end{proof}
The next step would be to find an approximation to $\widehat{h}(p,t)$.
\begin{lemma}\label{lem: approximation of h}
\begin{equation}\label{eq: approximation of h}
\begin{gathered}
\widehat{h}(p,t)=1-2\pi i t -2\pi^2 t^2 (\Sigma^2+1)+t^3 h(t) \\
-\frac{2\pi^2 |p|^2}{d}-\frac{\pi^2 |p|^2}{d}t h_1(t)+|p|^4h_2(p,t),
\end{gathered}
\end{equation}
where $h,h_1,h_2$ are analytic in the domain and there exist constants $M_0,M_1,M_2$, independent in $\delta$, such that 
\[|h(t)|\leq M_0+\frac{M_2}{\delta^2}.\]
\[|h_1(t)|\leq M_0+\frac{M_1}{\delta}.\]
\[|h_2(p,t)|\leq \left(M_0+\frac{M_1}{\delta} \right)M_{p.\delta},\]
with $M_{p,\delta}=\sup_{|x|\leq \frac{\pi^2 |p|^2}{d\delta}}|\phi(x)|$ and $\phi$ entire.
\end{lemma}
\begin{proof}
Using the exponential approximation we find that
\[\frac{e^{-\frac{\frac{\pi^2|p|^2}{d\delta}}{1+\frac{2\pi i t}{d\delta}}}}{\left(1+\frac{2\pi i t}{d\delta} \right)^{\frac{d}{2}}}=\frac{1}{\left(1+\frac{2\pi i t}{d\delta} \right)^{\frac{d}{2}}}-\frac{\pi^2|p|^2}{d\delta\left(1+\frac{2\pi i t}{d\delta} \right)^{\frac{d+2}{2}}}+\frac{\pi^4|p|^4}{d^2\delta^2\left(1+\frac{2\pi i t}{d\delta} \right)^{\frac{d+4}{2}}}\phi\left(\frac{\pi^2 |p|^2}{d\delta+2\pi i t} \right).\]
Another approximation we will need to use is the following:
\[ \frac{1}{\left(1+x \right)^\alpha}=1-\alpha x +\frac{\alpha(\alpha+1)}{2} x^2 +x^3 \cdot g_\alpha (x),\]
where $g_\alpha (x)$ is analytic in $|x|<1$.\\
We conclude that
\[\frac{1}{\left(1+\frac{2\pi i t}{d\delta} \right)^{\frac{d}{2}}}
=1-\frac{\pi i t}{\delta} -\frac{(d+2)}{4d\delta^2}\cdot 2\pi^2 t^2 -\frac{8\pi^3 i t^3}{d^3 \delta^3}g_{\frac{d}{2}}\left(\frac{2\pi i t}{d\delta} \right),\]
\[\frac{1}{\left(1+\frac{2\pi i t}{d\delta} \right)^{\frac{d+2}{2}}}
=1+\frac{(d+2)\pi i t}{d\delta}g_1\left(\frac{2\pi i t}{d\delta} \right),\]
\[\frac{1}{\left(1+\frac{2\pi i t}{d\delta} \right)^{\frac{d+4}{2}}}
=1+\frac{(d+4)\pi i t}{d\delta}g_2\left(\frac{2\pi i t}{d\delta} \right),\]
and so 
\[ \frac{\delta}{\left(1+\frac{2\pi i t}{d\delta} \right)^{\frac{d}{2}}}+\frac{(1-\delta)}{\left(1+\frac{2\pi i t}{d(1-\delta)} \right)^{\frac{d}{2}}}
=1-2\pi i t -2\pi^2 t^2 \left(\Sigma^2+1 \right)+\frac{8\pi^3 i t^3}{d^3} \left( \frac{g_{\frac{d}{2}}\left(\frac{2\pi i t}{d\delta} \right)}{\delta^2}+\frac{g_{\frac{d}{2}}\left(\frac{2\pi i t}{d(1-\delta)} \right)}{(1-\delta)^2} \right)\]
\[=1-2\pi i t -2\pi^2 t^2 \left(\Sigma^2+1 \right)+t^3 h(t),\]
where 
\[|h(t)|\leq \frac{8\pi^3}{d^3}\left(\frac{M_{sup}}{\delta^2}+\frac{M_{sup}}{(1-\delta)^2} \right),\] 
and $M_{sup}=\sup_{|x|<\frac{1}{2}}|g_{\frac{d}{2}}(x)|$.\\
Next, we see that
\[-\frac{\pi^2|p|^2}{d\left(1+\frac{2\pi i t}{d\delta} \right)^{\frac{d+2}{2}}}
=-\frac{\pi^2|p|^2}{d}\left(1+ \frac{(d+2)\pi i t}{d\delta}g_1\left(\frac{2\pi i t}{d\delta} \right) \right),\]
leading to
\[-\frac{\pi^2|p|^2}{d\left(1+\frac{2\pi i t}{d\delta} \right)^{\frac{d+2}{2}}}
-\frac{\pi^2|p|^2}{d\left(1+\frac{2\pi i t}{d(1-\delta)} \right)^{\frac{d+2}{2}}}=-\frac{2\pi^2|p|^2}{d}-\frac{\pi^2|p|^2}{d}\cdot t h_1(t),\]
with
\[|h_1(t)|\leq \frac{(d+2)\pi M_{1,sup}}{d\delta(1-\delta)},\]
and $M_{1,sup}=\sup_{|x|<\frac{1}{2}}|g_1(x)|$.\\
Lastly,

\[\frac{\pi^4|p|^4}{d^2\delta\left(1+\frac{2\pi i t}{d\delta} \right)^{\frac{d+4}{2}}}\phi\left(\frac{\pi^2 |p|^2}{d\delta+2\pi i t}\right)+\frac{\pi^4|p|^4}{d^2(1-\delta)\left(1+\frac{2\pi i t}{d(1-\delta)} \right)^{\frac{d+4}{2}}}\phi\left(\frac{\pi^2 |p|^2}{d(1-\delta)+2\pi i t}\right)\]
\[=|p|^4h_2(p,t),\]
where
\[|h_2(p,t)|\leq \frac{\pi^4}{d^2\delta}\phi\left(\frac{\pi^2|p|^2}{d\delta+2\pi i t}\right)+\frac{\pi^4}{d^2(1-\delta)}\phi\left(\frac{\pi^2 |p|^2}{d(1-\delta)+2\pi i t}\right) \]
\[
\leq \frac{\pi^4 M_{p,\delta}}{d^2\delta(1-\delta)},\]
and $M_{p,\delta}=\sup_{|x|\leq \frac{\pi^2 |p|^2}{d\delta}}|\phi(x)|$. 
The result follows readily from all the above estimations.
\end{proof}
Combining the two last Lemmas yields the following:
\begin{lemma}\label{lem: approximation of the difference}
When $|t|<\frac{d\delta^{1+\beta}}{4\pi}$ and $|p|\leq \delta^{\frac{1}{2}+\beta}$ we have that there exist constants $M_0,M_1,M_2$, independent of $\delta$, such that 
\begin{equation}\label{eq: approximation of the difference}
\begin{gathered}
\left\lvert \widehat{h}(p,t)-\widehat{\gamma_1}(p,t) \right\rvert \\
\leq |t|^3\left(M_0+\frac{M_1}{\delta}+\frac{M_2}{\delta^2}\right)+\frac{\pi^2|p|^2|t|}{d}\left(M_0+\frac{M_1}{\delta} \right)+|p|^4 \left(M_0+\frac{M_1}{\delta} \right).
\end{gathered}
\end{equation}
\end{lemma}
\begin{proof}
We can rewrite equation (\ref{eq: approximation of gamma_1}) as 
\[1-2\pi i t -2\pi^2 t^2 (\Sigma^2+1)+t^3 g(t)-\frac{2\pi^2|p|^2}{d}-\frac{\pi^2|p|^2}{d}t q_1(t)+|p|^4 f\left(|p|^2 \right) e^{-2\pi^2 \Sigma^2 t^2}e^{-2\pi i t},\]
with $q_1(t)=-4\pi i -4\pi^2 t \left(\Sigma^2 +1 \right)+2t^2 g(t)$. By the conditions on the domain and $g$ we know that $|q_1(t)|\leq M_0$ for some constant $M_0$.\\
Combining this with equation (\ref{eq: approximation of h}) and using the same notations as in the approximation Lemmas, we find that
\[\left\lvert \widehat{h}(p,t)-\widehat{\gamma_1}(p,t) \right\rvert\leq |t|^3\left(|g(t)|+|h(t)|\right)+\frac{\pi^2|p|^2|t|}{d}\left(|q_1(t)|+|h_1(t)| \right)\]
\[+|p|^4 \left(|f(p)|+|h_2(p,t)| \right).\]
Since $\eta=\delta^{\frac{1}{2}+\beta}$ we have that $M_{p,\delta}\leq \sup_{|x|\leq \frac{\pi^2 \delta^{2\beta}}{d}}|\phi(x)|\leq M_0$, and so we can find constants $M_0,M_1,M_2$ such that 
\[\left\lvert \widehat{h}(p,t)-\widehat{\gamma_1}(p,t) \right\rvert\leq |t|^3\left(M_0+\frac{M_1}{\delta}+\frac{M_2}{\delta^2}\right)+\frac{\pi^2|p|^2|t|}{d}\left(M_0+\frac{M_1}{\delta} \right)+|p|^4 \left(M_0+\frac{M_1}{\delta} \right),\]
which is the desired result.
\end{proof}
\begin{proof}[Proof of Theorem \ref{thm: t tiny p small case}]
Since
\[\iint_{|p|\leq\delta^{\frac{1}{2}+\beta} \times |t|\leq \frac{d\delta^{1+\beta}}{4\pi}}\left\lvert \widehat{h}(p,t)-\widehat{\gamma_1}(p,t) \right\rvert \left\lvert \widehat{h}(p,t)\right\rvert^{N-1} dpdt\]
\[ \leq \iint_{|p|\leq\delta^{\frac{1}{2}+\beta} \times |t|\leq \frac{d\delta^{1+\beta}}{4\pi}}\left\lvert \widehat{h}(p,t)-\widehat{\gamma_1}(p,t) \right\rvert  dpdt ,\]
inequality (\ref{eq: approximation of the difference}) shows that the above expression is bounded by
\[\left(M_0+\frac{M_1}{\delta}+\frac{M_2}{\delta^2} \right)\cdot \delta^{4+4\beta}\cdot \delta^{\frac{d}{2}+d\beta}+\left(M_0+\frac{M_1}{\delta} \right)\cdot \delta^{2+2\beta}\cdot \delta^{\frac{d}{2}+d\beta+1+2\beta}\]
\[+\left(M_0+\frac{M_1}{\delta} \right)\cdot \delta^{1+\beta}\cdot \delta^{\frac{d}{2}+d\beta+2+4\beta}\leq \frac{C_d}{\Sigma}\delta^{\frac{3}{2}+4\beta+\frac{d}{2}+d\beta}.\]
By Lemma \ref{lem: one dimensional gaussian estimation} we find that
\[\sum_{k=0}^{N-2}\int_{|t|\leq \frac{d\delta^{1+\beta}}{4\pi}}e^{-2\pi^2(N-k-1)\Sigma^2 t^2}dt=\sum_{k=1}^{N-1}\int_{|t|\leq \frac{d\delta^{1+\beta}}{4\pi}}e^{-2\pi^2k\Sigma^2 t^2}dt\leq \sqrt{\pi}\sum_{k=1}^{N-1} \frac{\sqrt{1-e^{-\frac{d^2 \Sigma^2 k \delta^{2+2\beta}}{4}}}}{\sqrt{2\pi^2 \Sigma^2 k }}\]
\[\leq \frac{C_d}{\Sigma}\sum_{k=1}^{N-1}\frac{1}{\sqrt{k}}\leq \frac{C_d\sqrt{N}}{\Sigma},\]
and since in our domain
\[\left\lvert \widehat{h}(p,t)-\widehat{\gamma_1}(p,t) \right\rvert
\leq C_d \delta^{1+3\beta}\]
we find that 
\[\sum_{k=0}^{N-2}\iint_{|p|\leq\delta^{\frac{1}{2}+\beta} \times |t|\leq \frac{d\delta^{1+\beta}}{4\pi}}\left\lvert \widehat{h}(p,t)-\widehat{\gamma_1}(p,t) \right\rvert \left\lvert \widehat{h}(p,t)\right\rvert^{k}\left\lvert \widehat{\gamma_1}(p,t)\right\rvert^{N-k-1} dpdt \]
\[ C_d \delta^{1+3\beta} \sum_{k=0}^{N-2}\iint_{|p|\leq\delta^{\frac{1}{2}+\beta} \times |t|\leq \frac{d\delta^{1+\beta}}{4\pi}}e^{-2\pi^2(N-k-1)\Sigma^2 t^2} dpdt
\leq \frac{C_d \sqrt{N}}{\Sigma}\delta^{1+3\beta+\frac{d}{2}+d\beta},\]
which finishes the proof.
\end{proof}
Now that we have all the domains sorted we can combine all the respective theorems into an appropriate approximation theorem. 

\subsection{The proof of the main approximation theorem}\label{subsec: proof of the main approximation theorem}
\begin{theorem}\label{thm: pre main approximation theorem}
For any $\beta>0$ and $0<\delta<\frac{1}{2}$ small enough we have that
\begin{equation}\label{eq: pre main approximation equation}
\begin{gathered}
\iint_{\mathbb{R}^d \times \mathbb{R}}\left\lvert \widehat{h}^N(p,t)-\widehat{\gamma_1}^N(p,t) \right\rvert dpdt \\
\leq \frac{C_d}{N^\frac{d-1}{2} \Sigma}\cdot e^{-\frac{d(d+2-4(1-\delta)\delta d)(N-2)\delta^{1+2\beta}}{128(1-\delta)}} \\
+\frac{NC_d}{\Sigma}\left( 1-\frac{d\delta^{1+2\beta}}{16}+\delta^{1+4\beta}\xi(\delta) \right)^{\frac{N}{2}}\cdot e^{-\frac{d(d+2-4d\delta(1-\delta)\delta^{1+2\beta}}{128(1-\delta)}} \\
+C_d \left(1-\frac{d\delta^{1+2\beta}}{16}+\delta^{1+4\beta}\xi(\delta) \right)^{N-5}
+C_d\delta^{1+2\beta}e^{-\frac{(N-2)\delta^{1+2\beta}}{4d}}
\\+ \frac{C_d}{\Sigma}\delta^{\frac{3}{2}+4\beta+\frac{d}{2}+d\beta}+ \frac{C_d \sqrt{N}}{\Sigma}\delta^{1+3\beta+\frac{d}{2}+d\beta}=\frac{\epsilon(N)}{\Sigma N^{\frac{d+1}{2}}},
\end{gathered}
\end{equation} 
where $C_d$ is a constant depending only on $d$, and $\xi$ is analytic in $|x|<\frac{1}{2}$.
\end{theorem}
\begin{proof}
This follows immediately from Theorems \ref{thm: large t case}, \ref{thm: small t large p case} and \ref{thm: t tiny p small case}.
\end{proof}

\begin{proof}[Proof of Theorem \ref{thm: main approximation theorem for the normalization function}]
We notice that the theorem is equivalent to showing that
\begin{equation}\label{eq: main approximation equation}
\sup_{v\in\mathbb{R}^d,u\in\mathbb{R}}\left\lvert h^{\ast N}(u,v)-\gamma_N(u,v)\right\rvert \leq \frac{\epsilon(N)}{\Sigma_{\delta_N} N^{\frac{d+1}{2}}}
\end{equation}
with $\lim_{N\rightarrow\infty}\epsilon(N)=0$.\\
Since 
\[\sup_{v\in\mathbb{R}^d,u\in\mathbb{R}}\left\lvert h^{\ast N}(u,v)-\gamma_N(u,v)\right\rvert \leq \iint_{\mathbb{R}^d \times \mathbb{R}}\left\lvert \widehat{h}^N(p,t)-\widehat{\gamma_1}^N(p,t) \right\rvert dpdt, \]
we only need to show that the specific choice of $\delta_N$ will give $\epsilon(N)$ that goes to zero, in the notations of Theorem \ref{thm: pre main approximation theorem}.\\
This will be true if we have the following conditions:
\begin{enumerate}[i.]
\item $\delta_N^{1+2\beta} N \underset{N\rightarrow\infty}{\longrightarrow}\infty$.
\item $N^{\frac{d+1}{2}}\delta_N^{\frac{3}{2}+4\beta+\frac{d}{2}+d\beta}\underset{N\rightarrow\infty}{\longrightarrow}0$.
\item $N^{\frac{d}{2}+1}\delta_N^{1+3\beta+\frac{d}{2}+d\beta}\underset{N\rightarrow\infty}{\longrightarrow}0$.
\end{enumerate}
The choice $\delta_N=\frac{1}{N^{1-\eta_\beta}}$ with 
\[\frac{2\beta}{1+2\beta}<\eta_\beta<\frac{(3+d)\beta}{1+3\beta+\frac{d}{2}+d\beta}\] will satisfy all the conditions.\\
Indeed,
\[\frac{N}{N^{(1-\eta)(1+2\beta)}}=N^{\eta(1+2\beta)-2\beta}.\]
Thus, in order to get the first condition we must have $\eta>\frac{2\beta}{1+2\beta}$.\\
Next we notice that
\[N^{\frac{d+1}{2}}N^{(\eta-1)\left(\frac{3}{2}+4\beta+\frac{d}{2}+d\beta\right)}=N^{\eta\left(\frac{3}{2}+4\beta+\frac{d}{2}+d\beta \right)-(1+4\beta+d\beta)},\]
so the second condition amounts to 
\[\eta<\frac{1+4\beta+d\beta}{\frac{3}{2}+4\beta+\frac{d}{2}+d\beta}\] 
which will obviously be satisfied for small enough $\beta$ and won't contradict the first one.\\
Lastly,
\[N^{\frac{d}{2}+1}N^{(\eta-1)\left(1+3\beta+\frac{d}{2}+d\beta\right)}=N^{\eta\left( 1+3\beta+\frac{d}{2}+d\beta \right)-(3+d)\beta},\]
so the third condition amounts to
\[\eta<\frac{(3+d)\beta}{1+3\beta+\frac{d}{2}+d\beta}.\]
In order to be consistent we must verify that 
\[\frac{2\beta}{1+2\beta}<\frac{(3+d)\beta}{1+3\beta+\frac{d}{2}+d\beta},\]
which is equivalent to 
\[2+6\beta+d+2d\beta<(3+d)(1+2\beta)=3+d+6\beta+2d\beta,\]
which is equivalent to $2<3$ and the proof is complete.
\end{proof}

\section{The main result}\label{sec: the main result}
We're finally ready to prove Theorem \ref{thm: main theorem}. The proof will consist of two theorems, one dealing with the denominator of (\ref{eq: entropic spectral gap dd}) and one with its numerator. Throughout this section the function $F_N$ will be defined as
\[F_N\left(v_1,\dots,v_N \right)=\frac{\prod_{i=1}^N f_{\delta_N}(v_i)}{\mathcal{Z}_B^N\left(f_{\delta_N},\sqrt{N},0 \right)}.\]
\begin{theorem}\label{thm: convergence of the entropy}
\begin{equation}\label{eq: convergence of the entropy}
\lim_{N\rightarrow\infty}\frac{H_N\left(F_N \right)}{N}=\frac{d\log 2}{2}.
\end{equation}
\end{theorem}
\begin{proof}
By the definition
\begin{equation}\label{eq: initial entropy expansion}
\begin{gathered}
H_N\left(F_N \right)=\frac{1}{\mathcal{Z}_B^N\left(f_{\delta_N},\sqrt{N},0 \right)}\int_{\mathcal{S}_B^N \left(N,0 \right)}\prod_{i=1}^N f_{\delta_N}(v_i) \log \left(\prod_{i=1}^N f_{\delta_N}(v_i) \right)d\sigma^N_{N,0} \\
-\log \left( \mathcal{Z}_B^N\left(f_{\delta_N},\sqrt{N},0 \right) \right)\\
=\frac{N}{\mathcal{Z}_B^N\left(f_{\delta_N},\sqrt{N},0 \right)}\int_{\mathcal{S}_B^N \left(N,0 \right)}\prod_{i=1}^N f_{\delta_N}(v_i) \log f_{\delta_N}(v_1)d\sigma^N_{N,0} \\
-\log \left( \mathcal{Z}_B^N\left(f_{\delta_N},\sqrt{N},0 \right) \right).
\end{gathered}
\end{equation}
Using Theorem \ref{thm: fubini on the boltzmann sphere} we find that
\[ \frac{1}{\mathcal{Z}_B^N\left(f_{\delta_N},\sqrt{N},0 \right)}\int_{\mathcal{S}_B^N \left(N,0 \right)}\prod_{i=1}^N f_{\delta_N}(v_i) \log f_{\delta_N}(v_1)d\sigma^N_{N,0}=\frac{\left\lvert \mathbb{S}^{d(N-2)-1} \right\rvert}{\left\lvert \mathbb{S}^{d(N-1)-1} \right\rvert}\cdot \frac{N^{\frac{d}{2}}}{(N-1)^{\frac{d}{2}}}\cdot \frac{1}{N^{\frac{d(N-1)-2}{2}}}\]
\[\int_{\Pi_{1,N}}dv_1\left(N-|v_1|^2-\frac{|v_1|^2}{N-1} \right)^{\frac{d(N-2)-2}{2}}
\cdot \frac{\mathcal{Z}_{N-1}\left(f_{\delta_N}, \sqrt{N-|v_1|^2},-v_1 \right)}{\mathcal{Z}_N\left(f_{\delta_N},\sqrt{N},0 \right)}.\]
At this point we notice that Theorem \ref{thm: main approximation theorem for the normalization function} can also be applied to $\mathcal{Z}_{N-1}$ with the appropriate changes. This leads us to conclude that 
\begin{equation}\label{eq: approximation of Z_N-1 for the entropy}
\begin{gathered}
\frac{\left\lvert \mathbb{S}^{d(N-2)-1}\right\rvert \left(N-|v_1|^2-\frac{|v_1|^2}{N-1} \right)^{\frac{d(N-2)-2}{2}}}{2(N-1)^{\frac{d}{2}}}\mathcal{Z}_{N-1}\left(f_{\delta_N},\sqrt{N-|v_1|^2},-v_1 \right) \\
=\frac{d^{\frac{d}{2}}}{\Sigma_{\delta_N} (N-1)^{\frac{d+1}{2}}(2\pi)^{\frac{d+1}{2}}}\left(e^{-\frac{d|v_1|^2}{2(N-1)}}e^{-\frac{\left(1-|v_1|^2\right)^2}{2\Sigma_{\delta_N}^2(N-1)}} +\lambda\left(\sqrt{N-|v_1|^2},-v_1 \right)\right),
\end{gathered}
\end{equation}
where $\sup_{v_1\in \Pi_{1,N}}\left\lvert \lambda\left(\sqrt{N-|v_1|^2,-v_1} \right) \right\rvert=\epsilon_1(N)\underset{N\rightarrow\infty}{\longrightarrow}0$.\\
Using Theorem \ref{thm: main approximation theorem for the normalization function} again we find that 
\begin{equation}\label{eq: approximation of Z_N for the entropy}
\frac{\left\lvert \mathbb{S}^{d(N-1)-1}\right\rvert N^{\frac{d(N-1)-2}{2}}}{2N^{\frac{d}{2}}}\mathcal{Z}_{N}\left(f_{\delta_N},\sqrt{N},0 \right)=\frac{d^{\frac{d}{2}}}{\Sigma_{\delta_N} N^{\frac{d+1}{2}}(2\pi)^{\frac{d+1}{2}}}\left(1+\epsilon(N)\right),
\end{equation}
where $\epsilon(N)\underset{N\rightarrow\infty}{\longrightarrow}0$. \\
Combining equations (\ref{eq: approximation of Z_N-1 for the entropy}) and (\ref{eq: approximation of Z_N for the entropy}) we have that 
\[\frac{1}{\mathcal{Z}_B^N\left(f_{\delta_N},\sqrt{N},0 \right)}\int_{\mathcal{S}_B^N \left(N,0 \right)}\prod_{i=1}^N f_{\delta_N}(v_i) \log f_{\delta_N}(v_1)d\sigma^N_{N,0}\]
\[=\left(\frac{N}{N-1} \right)^{\frac{d+1}{2}}\int_{\Pi_{1,N}}\frac{e^{-\frac{d|v_1|^2}{2(N-1)}}e^{-\frac{\left(1-|v_1|^2\right)^2}{2\Sigma_{\delta_N}^2(N-1)}} +\lambda\left(\sqrt{N-|v_1|^2,-v_1} \right)}{1+\epsilon(N)}f_{\delta_N}(v_1)\log f_{\delta_N}(v_1)dv_1.\]
Rewriting $f_{\delta_N}(v)=d^{\frac{d}{2}}\left(\frac{\delta_N^{\frac{d+2}{2}}}{\pi^{\frac{d}{2}}}e^{-d\delta_N|v|^2}+\frac{(1-\delta_N)^{\frac{d+2}{2}}}{\pi^{\frac{d}{2}}}e^{-d(1-\delta_N)|v|^2} \right)=d^{\frac{d}{2}}f_{1,N}(v)$ we find that $0<f_{1,N}<1$ and as such
\[ \left\lvert \chi_{\Pi_{1,N}}(v_1)\frac{e^{-\frac{d|v_1|^2}{2(N-1)}}e^{-\frac{\left(1-|v_1|^2\right)^2}{2\Sigma_{\delta_N}^2(N-1)}} +\lambda\left(\sqrt{N-|v_1|^2,-v_1} \right)}{1+\epsilon(N)}f_{\delta_N}(v_1)\log f_{\delta_N}(v_1)\right\rvert\]
\[\leq \frac{1+|\epsilon_1(N)|}{1-|\epsilon(N)|}\left(\frac{d\log d}{2} f_{\delta_N}(v_1) - f_{\delta_N}(v_1)\log f_{1,N}(v_1) \right)\]
\[=\frac{1+|\epsilon_1(N)|}{1-|\epsilon(N)|}\left(d\log d \cdot f_{\delta_N}(v_1) - f_{\delta_N}(v_1)\log f_{\delta_N}(v_1) \right)\]
\[\leq \frac{1+|\epsilon_1(N)|}{1-|\epsilon(N)|}\cdot d\log d \cdot f_{\delta_N}(v_1)\]
\[ -\frac{1+|\epsilon_1(N)|}{1-|\epsilon(N)|}\left(\delta_N M_{\frac{1}{2d\delta_N}}(v_1)\log \left(\delta_N M_{\frac{1}{2d\delta_N}}(v_1) \right)+(1-\delta_N) M_{\frac{1}{2d(1-\delta_N)}}(v_1)\log \left((1-\delta_N) M_{\frac{1}{2d(1-\delta_N)}}(v_1) \right) \right)\]
\[= \frac{1+|\epsilon_1(N)|}{1-|\epsilon(N)|}\cdot d\log d \cdot f_{\delta_N}(v_1)\]
\[+\frac{1+|\epsilon_1(N)|}{1-|\epsilon(N)|}\cdot\delta_N M_{\frac{1}{2d\delta_N}}(v_1)\left(d\delta_N|v_1|^2-\frac{d}{2}\log \left(\frac{d}{\pi}\right)-\frac{d+2}{2}\log(\delta_N)\right)  \]
\[\frac{1+|\epsilon_1(N)|}{1-|\epsilon(N)|}\cdot(1-\delta_N) M_{\frac{1}{2d(1-\delta_N)}}(v_1)\left(d(1-\delta_N)|v_1|^2-\frac{d}{2}\log \left(\frac{d}{\pi}\right)-\frac{d+2}{2}\log(1-\delta_N) \right)\]
\[=g_N(v_1).\]
We notice that $g_N(v_1)\underset{N\rightarrow\infty}{\longrightarrow}\left(\frac{d\log d}{2}+\frac{d\log\pi}{2}+d|v_1|^2\right) M_{\frac{1}{2d}}(v_1)$ pointwise and
\[\int_{\mathbb{R}^d}g_N(v_1)dv_1=\frac{1+|\epsilon_1(N)|}{1-|\epsilon(N)|} \Bigg( d\log d +\frac{d\delta_N}{2}-\frac{d\delta_N}{2}\log\left(\frac{d}{\pi}\right)-\frac{(d+2)\delta_N \log(\delta_N)}{2}  \]
\[ +\frac{d(1-\delta_N)}{2}-\frac{d(1-\delta_N)}{2}\log\left(\frac{d}{\pi}\right)-\frac{(d+2)(1-\delta_N) \log((1-\delta_N)}{2} \Bigg). \]
Thus
\[\lim_{N\rightarrow\infty}\int_{\mathbb{R}^d}g_N(v_1)dv_1=\frac{d \log d}{2}+\frac{d\log \pi}{2}+\frac{d}{2}=\int_{\mathbb{R}^d}\lim_{N\rightarrow\infty}g_N(v_1)dv_1.\]
Since clearly
\[ \chi_{\Pi_{1,N}}(v_1)\frac{e^{-\frac{d|v_1|^2}{2(N-1)}}e^{-\frac{\left(1-|v_1|^2\right)^2}{2\Sigma_{\delta_N}^2(N-1)}} +\lambda\left(\sqrt{N-|v_1|^2,-v_1} \right)}{1+\epsilon(N)}f_{\delta_N}(v_1)\log f_{\delta_N}(v_1)\]
\[\underset{N\rightarrow\infty}{\longrightarrow}M_{\frac{1}{2d}}(v_1)\log \left( M_{\frac{1}{2d}}(v_1) \right),\]
we conclude by the Generalised Dominated Convergence Theorem that 
\begin{equation}\label{eq: main step in entropy convergence}
\begin{gathered}
\lim_{N\rightarrow\infty}\frac{1}{\mathcal{Z}_B^N\left(f_{\delta_N},\sqrt{N},0 \right)}\int_{\mathcal{S}_B^N \left(N,0 \right)}\prod_{i=1}^N f_{\delta_N}(v_i) \log f_{\delta_N}(v_1)d\sigma^N_{N,0}\\
=\int_{\mathbb{R}^d}M_{\frac{1}{2d}}(v)\log \left( M_{\frac{1}{2d}}(v) \right)dv
=\frac{d}{2}\log d - \frac{d}{2} \log \pi -\frac{d}{2}.
\end{gathered}
\end{equation}

We're only left with the evaluation the term $\log\left(\mathcal{Z}_N\left(f_{\delta_N,\sqrt{N},0} \right) \right)$ to complete the proof. Using (\ref{eq: approximation of Z_N for the entropy}) along with $\left\lvert \mathbb{S}^{m-1} \right\rvert=\frac{2\pi^{\frac{m}{2}}}{\Gamma\left(\frac{m}{2} \right)}$ and an approximation for the gamma function yields
\[\mathcal{Z}_N\left(f_{\delta_N},\sqrt{N},0 \right)=\frac{2d^{\frac{d}{2}}(1+\epsilon_2(N))}{(2\pi)^{\frac{d+1}{2}}\Sigma_{\delta_N}N^{\frac{d(N-1)-1}{2}}\left\lvert \mathbb{S}^{d(N-1)-1} \right\rvert}\]
\[=\frac{d^{\frac{d}{2}}\pi^{-\frac{dN}{2}}\Gamma\left(\frac{d(N-1)}{2} \right)(1+\epsilon_2(N))}{2^{\frac{d+1}{2}}\sqrt{\pi}\Sigma_{\delta_N}N^{\frac{d(N-1)-1}{2}}}\]
\[=\frac{d^{\frac{d}{2}}\pi^{-\frac{dN}{2}}\left(\left( \frac{d(N-1)}{2} \right)^{\frac{d(N-1)-1}{2}}e^{-\frac{d(N-1)}{2}}\sqrt{2\pi}(1+\epsilon_3(N)) \right)(1+\epsilon_2(N))}{2^{\frac{d+1}{2}}\sqrt{\pi}\Sigma_{\delta_N}N^{\frac{d(N-1)-1}{2}}} \]
\[=\left(\frac{de}{2} \right)^{\frac{d}{2}}\cdot \frac{(\pi e)^{-\frac{dN}{2}}}{\Sigma_{\delta_N}}\cdot\left(\frac{d}{2}\left(1-\frac{1}{N} \right) \right)^{\frac{d(N-1)-1}{2}}\cdot(1+\epsilon(N)). \]
Thus,
\begin{equation}\label{eq: secondary step in entropy convergence}
\begin{gathered}
\lim_{N\rightarrow\infty}\frac{\log\left( \mathcal{Z}_N\left(f_{\delta_N},\sqrt{N},0 \right)\right)}{N}=-\frac{d}{2}\log(\pi e)+\frac{d}{2}\log\left(\frac{d}{2} \right)\\
=\frac{d\log d}{2}-\frac{d\log \pi}{2}-\frac{d}{2}-\frac{d\log 2}{2}.
\end{gathered}
\end{equation}
Combining (\ref{eq: initial entropy expansion}), (\ref{eq: main step in entropy convergence}) and (\ref{eq: secondary step in entropy convergence}) yields the result.
\end{proof}

\begin{theorem}\label{thm: entropy production estimation}
There exists a constant $C_\delta$, depending only on the behaviour of $\delta$ such that 
\begin{equation}\label{eq: entropy production estimation}
\left\langle \log F_N,(I-Q)F_N \right\rangle \leq -C_\delta \delta_N \log \delta_N.
\end{equation}
\end{theorem}
\begin{proof}
Since $\left\langle C,(I-Q)F_N \right\rangle=0$ for any constant $C$, and with the same notation of the proof of Theorem \ref{thm: entropy production estimation}, we find that 
\[\left\langle \log F_N,(I-Q)F_N \right\rangle=\frac{2}{N(N-1)\mathcal{Z}_N\left(f_{\delta_N},\sqrt{N},0)\right)} \sum_{i<j}\int_{\mathcal{S}_B^N(N,0)}d\sigma^N_{N,0}\log \left(\prod_{k=1}^N f_{1,N}(v_k)\right)\]
\[\int_{\mathbb{S}^{d-1}}\left[f_{\delta_N}^{\otimes N}\left(v_1,\dots,v_i,\dots,v_j,\dots ,v_N \right)-f_{\delta_N}^{\otimes N}\left(v_1,\dots,v_i(\omega),\dots,v_j(\omega),\dots ,v_N \right)\right]d\omega\]
\[=\frac{2}{N(N-1)\mathcal{Z}_N \left(f_{\delta_N},\sqrt{N},0\right)}\sum_{i<j}\sum_{k=1}^N\int_{\mathcal{S}_B^N(N,0)}d\sigma^N_{N,0} \log \left(f_{1,N}(v_k)\right)\]
\[\int_{\mathbb{S}^{d-1}}\left[f_{\delta_N}^{\otimes N}\left(v_1,\dots,v_i,\dots,v_j,\dots ,v_N \right)-f_{\delta_N}^{\otimes N}\left(v_1,\dots,v_i(\omega),\dots,v_j(\omega),\dots ,v_N \right)\right]d\omega.\]
We notice that if $k\not=i,j$ then the integral is equal to
\[\frac{\left\lvert \mathbb{S}^{d(N-2)-1} \right\rvert}{\left\lvert \mathbb{S}^{d(N-1)-1} \right\rvert}\cdot \frac{N^{\frac{d}{2}}}{(N-1)^\frac{d}{2} N^{\frac{d(N-1)-2}{2}}}\int_{\mathbb{S}^{d-1}}d\omega \int_{v_k\in\Pi_{1,N}}\log \left(f_{1,N}(v_k)\right)\left(N-|v_k|^2-\frac{\left\lvert v_k \right\rvert^2}{N-1} \right)^{\frac{d(N-2)-2}{2}}\]
\[\int_{\mathcal{S}_B^{N-1}\left(N-|v_k|^2,-v_k\right)}\Bigg[f_{\delta_N}^{\otimes N}\left(v_1,\dots,v_i,\dots,v_j,\dots ,v_N \right)\]
\[-f_{\delta_N}^{\otimes N}\left(v_1,\dots,v_i(\omega),\dots,v_j(\omega),\dots ,v_N \right)\Bigg]d\sigma^{N-1}_{N-|v_k|^2,-v_k}=0,\]
due to the symmetry of the Boltzmann sphere. Also, we see that 
\[\left\langle \log F_N,(I-Q)F_N \right\rangle=\frac{2}{N(N-1)\mathcal{Z}_N\left(f_{\delta_N},\sqrt{N},0\right)}\]
\[\sum_{i<j}\int_{\mathcal{S}_B^N(N,0)}d\sigma^N_{N,0}\left(\log \left( f_{1,N}(v_i)\right)+\log \left( f_{1,N}(v_j)\right)\right)\]
\[\int_{\mathbb{S}^{d-1}}\left[f_{\delta_N}^{\otimes N}\left(v_1,\dots,v_i,\dots,v_j,\dots ,v_N \right)-f_{\delta_N}^{\otimes N}\left(v_1,\dots,v_i(\omega),\dots,v_j(\omega),\dots ,v_N \right)\right]d\omega\]
\[=\frac{2}{N(N-1)\mathcal{Z}_N\left(f_{\delta_N},\sqrt{N},0\right)}\sum_{i\not=j}\int_{\mathcal{S}_B^N(N,0)}d\sigma^N_{N,0}\log \left( f_{1,N}(v_i)\right)\]
\[\int_{\mathbb{S}^{d-1}}\left[f_{\delta_N}^{\otimes N}\left(v_1,\dots,v_i,\dots,v_j,\dots ,v_N \right)-f_{\delta_N}^{\otimes N}\left(v_1,\dots,v_i(\omega),\dots,v_j(\omega),\dots ,v_N \right)\right]d\omega\]
\[=\frac{2}{\mathcal{Z}_N\left(f_{\delta_N},\sqrt{N},0\right)}\]
\[\int_{\mathcal{S}_B^N(N,0)}d\sigma^N_{N,0}\log \left( f_{1,N}(v_1)\right)\int_{\mathbb{S}^{d-1}}\left[f_{\delta_N}^{\otimes N}\left(v_1,v_2,\dots ,v_N \right)-f_{\delta_N}^{\otimes N}\left(v_1(\omega),v_2(\omega),\dots,v_N \right)\right]d\omega\]
\[=\frac{2}{\mathcal{Z}_N\left(f_{\delta_N},\sqrt{N},0\right)}\int_{\mathbb{S}^{d-1}}d\omega \frac{\left\lvert \mathbb{S}^{d(N-3)-1} \right\rvert}{\left\lvert \mathbb{S}^{d(N-1)-1} \right\rvert}\cdot \frac{N^{\frac{d}{2}}}{(N-2)^\frac{d}{2} N^{\frac{d(N-1)-2}{2}}}\]
\[ \int_{\Pi_{2,N}}dv_1dv_2\left(N-\left(|v_1|^2+|v_2|^2 \right)-\frac{\left\lvert v_1+v_2\right\rvert^2}{N-2} \right)^{\frac{d(N-3)-2}{2}}\log \left(f_{1,N}(v_1) \right)\]
\[\left(f_{\delta_N}(v_1)f_{\delta_N}(v_2)-f_{\delta_N}(v_1(\omega))f_{\delta_N}(v_2(\omega)) \right)\mathcal{Z}_{N-2}\left(f_{\delta_N},\sqrt{N-\left(|v_1|^2+|v_2|^2 \right)}, -v_1-v_2 \right).\]
Using Theorem \ref{thm: main approximation theorem for the normalization function} for $\mathcal{Z}_{N-2}$ (with the appropriate changes) gives us 
\begin{equation}\label{eq: approximation of Z_N-2 for the entropy}
\begin{gathered}
\frac{\left\lvert \mathbb{S}^{d(N-3)-1}\right\rvert \left(N-\left(|v_1|^2+|v_2|^2 \right)-\frac{|v_1+v_2|^2}{N-2} \right)^{\frac{d(N-3)-2}{2}}}{2(N-2)^{\frac{d}{2}}}\\
\\
\mathcal{Z}_{N-2}\left(f_{\delta_N},\sqrt{N-\left(|v_1|^2+|v_2|^2 \right)},-v_1-v_2 \right)=\frac{d^{\frac{d}{2}}}{\Sigma_{\delta_N} (N-2)^{\frac{d+1}{2}}(2\pi)^{\frac{d+1}{2}}} \\
\left(e^{-\frac{d|v_1+v_2|^2}{2(N-2)}}e^{-\frac{\left(2-|v_1|^2-|v_2|^2\right)^2}{2\Sigma_{\delta_N}^2(N-2)}} +\lambda\left(\sqrt{N-\left(|v_1|^2+|v_2|^2\right)},-v_1-v_2 \right)\right),
\end{gathered}
\end{equation}
where $\sup_{v_1,v_2\in \Pi_{2,N}} \left\lvert \lambda\left(\sqrt{N-\left(|v_1|^2+|v_2|^2\right)},-v_1-v_2 \right) \right\rvert=\epsilon_1(N)\underset{N\rightarrow\infty}{\longrightarrow}0$.\\
Plugging (\ref{eq: approximation of Z_N-2 for the entropy}) and (\ref{eq: approximation of Z_N for the entropy}) into our equation we find that
\[\left\langle \log F_N,(I-Q)F_N \right\rangle=2\left(\frac{N}{N-2}\right)^{\frac{d+1}{2}}\int_{\mathbb{S}^{d-1}}\int_{\Pi_{2,N}}dv_1dv_2 d\omega\]
\[ \frac{e^{-\frac{d|v_1+v_2|^2}{2(N-2)}}e^{-\frac{\left(2-|v_1|^2-|v_2|^2\right)^2}{2\Sigma_{\delta_N}^2(N-2)}} +\lambda\left(\sqrt{N-\left(|v_1|^2+|v_2|^2\right)},-v_1-v_2 \right)}{1+\epsilon(N))}\]
\[\log \left(f_{1,N}(v_1) \right)\left(f_{\delta_N}(v_1)f_{\delta_N}(v_2)-f_{\delta_N}(v_1(\omega))f_{\delta_N}(v_2(\omega)) \right).\]
At this point we notice that since $|v_1|^2+|v_2|^2=|v_1(\omega)|^2+|v_2(\omega)|^2$ and $v_1+v_2=v_1(\omega)+v_2(\omega)$ the domain $\Pi_{2,N}$ is symmetric to changing $1$ with $2$ and $v$ with $v(\omega)$. Thus we can rewrite the above as
\[\left\langle \log F_N,(I-Q)F_N \right\rangle=\frac{1}{2}\left(\frac{N}{N-2}\right)^{\frac{d+1}{2}}\int_{\mathbb{S}^{d-1}}\int_{\Pi_{2,N}}dv_1dv_2 d\omega\]
\[ \frac{e^{-\frac{d|v_1+v_2|^2}{2(N-2)}}e^{-\frac{\left(2-|v_1|^2-|v_2|^2\right)^2}{2\Sigma_{\delta_N}^2(N-2)}} +\lambda\left(\sqrt{N-\left(|v_1|^2+|v_2|^2\right)},-v_1-v_2 \right)}{1+\epsilon(N))}\]
\[\log \left(\frac{f_{1,N}(v_1)f_{1,N}(v_2)}{f_{1,N}(v_1(\omega))f_{1,N}(v_2(\omega))} \right)\left(f_{\delta_N}(v_1)f_{\delta_N}(v_2)-f_{\delta_N}(v_1(\omega))f_{\delta_N}(v_2(\omega)) \right),\]
whose integrand is clearly non-negative. As such
\[\left\langle \log F_N,(I-Q)F_N \right\rangle\leq\frac{1}{2}\left(\frac{N}{N-2}\right)^{\frac{d+1}{2}}\int_{\mathbb{S}^{d-1}}\int_{\mathbb{R}^{2d}}dv_1dv_2 d\omega\]
\[ \frac{e^{-\frac{d|v_1+v_2|^2}{2(N-2)}}e^{-\frac{\left(2-|v_1|^2-|v_2|^2\right)^2}{2\Sigma_{\delta_N}^2(N-2)}} +\lambda\left(\sqrt{N-\left(|v_1|^2+|v_2|^2\right)},-v_1-v_2 \right)}{1+\epsilon(N))}\]
\[\log \left(\frac{f_{1,N}(v_1)f_{1,N}(v_2)}{f_{1,N}(v_1(\omega))f_{1,N}(v_2(\omega))} \right)\left(f_{\delta_N}(v_1)f_{\delta_N}(v_2)-f_{\delta_N}(v_1(\omega))f_{\delta_N}(v_2(\omega)) \right)\]
\[=2\left(\frac{N}{N-2}\right)^{\frac{d+1}{2}}\int_{\mathbb{S}^{d-1}}\int_{\mathbb{R}^{2d}}dv_1dv_2 d\omega\]
\[ \frac{e^{-\frac{d|v_1+v_2|^2}{2(N-2)}}e^{-\frac{\left(2-|v_1|^2-|v_2|^2\right)^2}{2\Sigma_{\delta_N}^2(N-2)}} +\lambda\left(\sqrt{N-\left(|v_1|^2+|v_2|^2\right)},-v_1-v_2 \right)}{1+\epsilon(N))}\]
\[\log \left(f_{1,N}(v_1) \right)\left(f_{\delta_N}(v_1)f_{\delta_N}(v_2)-f_{\delta_N}(v_1(\omega))f_{\delta_N}(v_2(\omega)) \right)\]
\[\leq \frac{2(1+|\epsilon_1(N))|}{1-|\epsilon(N))|}\left(\frac{N}{N-2}\right)^{\frac{d+1}{2}}\int_{\mathbb{S}^{d-1}}\int_{\mathbb{R}^{2d}}dv_1dv_2 d\omega\]
\[\left\lvert \log \left(f_{1,N}(v_1) \right)\right\rvert \left\lvert \left(f_{\delta_N}(v_1)f_{\delta_N}(v_2)-f_{\delta_N}(v_1(\omega))f_{\delta_N}(v_2(\omega)) \right)\right\rvert,\]
and since $0<f_{1,N}<1$ we conclude that 
\begin{equation}\label{eq: entropy production main evaluation}
\begin{gathered}
\left\langle \log F_N,(I-Q)F_N \right\rangle \leq \frac{2(1+|\epsilon_1(N))|}{1-|\epsilon(N))|}\left(\frac{N}{N-2}\right)^{\frac{d+1}{2}} \int_{\mathbb{S}^{d-1}}\int_{\mathbb{R}^{2d}}dv_1 dv_2 d\omega \\
\left(-\log \left(f_{1,N}(v_1)\right) \right)\left\lvert \left(f_{\delta_N}(v_1)f_{\delta_N}(v_2)-f_{\delta_N}(v_1(\omega))f_{\delta_N}(v_2(\omega)) \right)\right\rvert.
\end{gathered}
\end{equation}
Next, we notice that
\begin{equation}\label{eq: estimation of -log f_1}
\begin{gathered}
-\log \left(f_{1,N}(v_1)\right) \leq -\log \left(\frac{\delta_N^{\frac{d+2}{2}}}{\pi^{\frac{d}{2}}}e^{-d\delta_N |v_1|^2} \right)\\
\leq d\delta_N \left(|v_1|^2+|v_2|^2\right)+\frac{d\log \pi}{2}-\frac{d+2}{2}\log(\delta_N).
\end{gathered}
\end{equation}
Also, since
\[f_{\delta_N}(v_1)f_{\delta_N}(v_2) 
=\delta_N^2 M_{\frac{1}{2d\delta_N}}(v_1)M_{\frac{1}{2d\delta_N}}(v_2)\]
\[+\delta_N(1-\delta_N)\left(M_{\frac{1}{2d\delta_N}}(v_1)M_{\frac{1}{2d(1-\delta_N)}}(v_2)+M_{\frac{1}{2d(1-\delta_N)}}(v_1)M_{\frac{1}{2d\delta_N}}(v_2) \right)\]
\[(1-\delta_N)^2 M_{\frac{1}{2d(1-\delta_N)}}(v_1)M_{\frac{1}{2d(1-\delta_N)}}(v_2),\]
we find that 
\begin{equation}\label{eq: estimation of the double difference term}
\begin{gathered}
\left\lvert \left(f_{\delta_N}(v_1)f_{\delta_N}(v_2)-f_{\delta_N}(v_1(\omega))f_{\delta_N}(v_2(\omega)) \right)\right\rvert \\
\leq \delta_N(1-\delta_N)\Bigg(M_{\frac{1}{2d\delta_N}}(v_1)M_{\frac{1}{2d(1-\delta_N)}}(v_2)+M_{\frac{1}{2d(1-\delta_N)}}(v_1)M_{\frac{1}{2d\delta_N}}(v_2)\\
+M_{\frac{1}{2d\delta_N}}(v_1(\omega))M_{\frac{1}{2d(1-\delta_N)}}(v_2(\omega))+M_{\frac{1}{2d(1-\delta_N)}}(v_1(\omega))M_{\frac{1}{2d\delta_N}}(v_2(\omega))\Bigg).
\end{gathered}
\end{equation}
Plugging (\ref{eq: estimation of -log f_1}) and (\ref{eq: estimation of the double difference term}) into (\ref{eq: entropy production main evaluation}) and using symmetry we find that 
\[\left\langle \log F_N,(I-Q)F_N \right\rangle \leq \frac{8(1+|\epsilon_1(N))|}{1-|\epsilon(N))|}\left(\frac{N}{N-2}\right)^{\frac{d+1}{2}}\delta_N(1-\delta_N)\]
\[ \int_{\mathbb{R}^{2d}} \left( d\delta_N \left(|v_1|^2+|v_2|^2\right)+\frac{d\log \pi}{2}-\frac{d+2}{2}\log(\delta_N) \right) M_{\frac{1}{2d\delta_N}}(v_1)M_{\frac{1}{2d(1-\delta_N)}}(v_2)dv_1dv_2\]
\[=\frac{8(1+|\epsilon_1(N))|}{1-|\epsilon(N))|}\left(\frac{N}{N-2}\right)^{\frac{d+1}{2}}\delta_N(1-\delta_N)\left(\frac{d\log \pi}{2}+\frac{d}{2(1-\delta_N)}-\frac{d+2}{2}\log\delta_N \right),\]
which proves the result.
\end{proof}

\begin{proof}[Proof of Theorem \ref{thm: main theorem}]
With the same family of functions as in Theorems \ref{thm: convergence of the entropy} and \ref{thm: entropy production estimation} we find that 
\[\Gamma_N \leq \frac{\left\langle \log F_N, N(I-Q)F_N \right\rangle}{H_N(F_N)}=\frac{\left\langle \log F_N, (I-Q)F_N \right\rangle}{\frac{H(F_N)}{N}}\leq -C_\delta \delta_N \log \delta_N ,\]
and plugging $\delta_N=\frac{1}{N^{1-\eta}}$, with $\eta$ satisfying the conditions of Theorem \ref{thm: main approximation theorem for the normalization function}, for an arbitrary $\beta>0$, yields the result.
\end{proof}

\section{Final Remarks}\label{sec: final remarks}
In this paper we managed to see that the addition of more dimensions, allowing conservation of momentum as well as energy, doesn't help the entropy-entropy production ratio. Nor does it worsen it. Moreover, it is not difficult to see that Theorem \ref{thm: main theorem} can be extended to a more general case of collisions operators. Indeed, if we define
\[Q_{\gamma}F=\frac{2}{N(N-1)}\sum_{i<j}\]
\[\int_{\mathbb{S}^{d-1}}B_{\gamma}(v_i,v_j)F\left(v_1,\dots,v_i(\omega),\dots,v_j(\omega),\dots,v_N\right)d\sigma^d ,\]
where $B_{\gamma}(v_i,v_j)$ is an appropriate positive function depending on $|v_i|^2+|v_j|^2$ and $v_i+v_j$, to conserve the symmetry of the problem (compare with (\ref{eq: collision operator})), then we see that in the case when 
\[B_{\gamma}(v_i,v_j)\leq \left( 1 + |v_i|^2+|v_j|^2 \right)^{\frac{\gamma}{2}},\]
or 
\[B_{\gamma}(v_i,v_j)\leq \left\lvert v_i-v_j \right\rvert^{\gamma},\]
we get that 
\[ \Gamma^{\gamma}_N \leq C_d N^{\frac{\gamma}{2}}\Gamma_N, \]
where $C_d$ is a constant depending only on $d$ and $\Gamma^{\gamma}_N$ is defined as (\ref{eq: entropic spectral gap dd}) but with $Q_\gamma$ replacing $Q$ in the definition of $D(F_N)$. Thus, we can conclude that
\begin{theorem}
For any $0<\eta<1$ there exists a constant $C_\eta$, depending only on $\eta$, such that $\Gamma_N$, defined in (\ref{eq: entropic spectral gap dd}), satisfies
\begin{equation}\label{eq: non trivial collision kernel entropy entropy production}
\Gamma^{\gamma}_N \leq \frac{C_\eta}{N^{\eta-\frac{\gamma}{2}}}.
\end{equation}
\end{theorem}
Possible questions that should be considered in the future, even in the one dimensional case, are:
\begin{itemize}
\item For our specific choice of 'generating function', $f_{\delta_N}$, we notice that the fourth moment, connected to $\Sigma_{\delta_N}^2$, explodes as $N$ goes to infinity. Would restricting such behaviour result in a better ratio?
\item Intuitively speaking, a reason for such 'slow relaxation' lies in the fact that we're trying to equilibrate many 'stable' states (represented by the Maxwellian with parameter $1/\left(2(1-\delta_N)\right)$) with very few highly energetic states (represented by the Maxwellian with parameter $1/(2\delta_N)$). Will restricting our class of function to one where the velocities are 'close' in some sense result in a better ratio?
\end{itemize}
Another question that can be asked in the multi dimensional case is the following:
\begin{itemize}
\item Can one extend Villani's proof in \cite{Villani} to the $d$-dimensional case?
\end{itemize}
While we have no answers to any of the above so far, we're hoping that some of the presented questions will be solved, for the one dimensional case as well as for $d$-dimensions. 

\appendix
\section{Additional Proofs}\label{app: additional proofs}
This Appendix contains several proofs of Lemmas that would have encumbered the main article, but pose a necessary step in the proof of our main result.
\begin{theorem}\label{thm: fubini on the boltzmann sphere app}
\begin{equation}\nonumber
\begin{gathered}
\int_{\mathcal{S}_B^N(E,z)}Fd\sigma^N_{E,z}=\frac{\left\lvert \mathbb{S}^{d(N-j-1)-1} \right\rvert}{\left\lvert \mathbb{S}^{d(N-1)-1} \right\rvert}\cdot \frac{N^{\frac{d}{2}}}{(N-j)^\frac{d}{2} \left(E-\frac{|z|^2}{N}\right)^{\frac{d(N-1)-2}{2}}} \\
\int_{\Pi_j(E,z)}dv_1\dots dv_j\left(E-\sum_{i=1}^j |v_i|^2-\frac{\left\lvert z-\sum_{i=1}^j v_i \right\rvert^2}{N-j} \right)^{\frac{d(N-j-1)-2}{2}}\\
\int_{\mathcal{S}_B^{N-j}\left(E-\sum_{i=i}^j|v_i|^2,z-\sum_{i=1}^j v_i\right)}F d\sigma^{N-j}_{E-\sum_{i=i}^j|v_i|^2,z-\sum_{i=1}^j v_i},
\end{gathered}
\end{equation}
where $\Pi_j(E,z)=\left\lbrace \sum_{i=1}^j |v_i|^2 + \frac{|z-\sum_{i=1}^j v_i|^2}{N-j} \leq E \right\rbrace$.
\end{theorem}

\begin{proof}
The proof relies heavily on the transformation (\ref{eq: boltzmann sphere transformation}) and the following Fubini-like formula for spheres (which can be found in \cite{Einav}):
\begin{equation}\label{eq: fubini theorem on the sphere}
\begin{gathered}
\int_{\mathbb{S}^{m-1}(r)}fd\gamma^m_r=\frac{\left\lvert \mathbb{S}^{m-j-1} \right\rvert}{\left\lvert \mathbb{S}^{m-1} \right\rvert r^{m-2}}\\
\int_{\sum_{i=1}^j x_i^2 \leq r^2}dx_1 \dots dx_j \left(r^2-\sum_{i=1}^j x_i^2 \right)^{\frac{m-j-2}{2}}
\int_{\mathbb{S}^{m-j-1}\left(\sqrt{r^2-\sum_{i=1}^j x_i^2} \right)}fd\gamma^{m-j}_{\sqrt{r^2-\sum_{i=1}^j x_i^2}},
\end{gathered}
\end{equation}
where $d\gamma^m_r$ is the uniform probability measure on the appropriate sphere.\\
We start by defining the new variables
\[\left(\xi_1,\dots,\xi_j \right)=R_1 \left(v_1,\dots,v_j \right),\]
\[\left(\xi_{j+1},\dots,\xi_N \right)=R_2 \left(v_{j+1},\dots,v_N \right),\]
where $R_1,R_2$ are transformation like (\ref{eq: boltzmann sphere transformation}). We notice that under the above transformation the domain 
\[\sum_{i=1}^N |v_i|^2=E \qquad \sum_{i=1}^N v_i=z\]
transforms into
\[\sum_{i=1}^N |\xi_i|^2=E \qquad \sqrt{j}\xi_j+\sqrt{N-j}\xi_N=z,\]
which can be written as
\begin{equation}\label{eq: pre domain in xi}
\sum_{i=1}^{N-1} |\xi_i|^2+\frac{1}{N-j}\left\lvert z-\sqrt{j}\xi_j \right\rvert^2=E.
\end{equation}
The following computation:
\[|x|^2+\frac{1}{N-j}\left(z-\sqrt{j}x \right)^2=|x|^2+\frac{1}{N-j}\left( |z|^2-2\sqrt{j}zx+j|x|^2 \right)\]
\[=\frac{1}{N-j}\left( |z|^2-2\sqrt{j}zx+N|x|^2 \right)
=\frac{1}{N-j}\left(N\left(x-\frac{\sqrt{j}z}{N} \right)^2+\frac{(N-j)|z|^2}{N} \right),\]
shows that (\ref{eq: pre domain in xi}) is
\begin{equation}\label{eq: domain in xi}
\sum_{i=1, \; i\not=j}^{N-1} |\xi_i|^2+\frac{N}{N-j}\left(\xi_j-\frac{\sqrt{j}z}{N} \right)^2=E-\frac{|z|^2}{N} .
\end{equation}
Denoting by $\widetilde{\xi_j}=\sqrt{\frac{N}{N-j}}\left(\xi_j-\frac{\sqrt{j}z}{N} \right)$ and using the fact that $R=R_1\otimes R_2$ is orthogonal along with (\ref{eq: fubini theorem on the sphere}) we find that
\[\int_{\mathcal{S}_B^N(E,z)}f d\sigma^N_{E,z}=\int_{\sum_{i=1, \; i\not=j}^{N-1} |\xi_i|^2+|\widetilde{\xi_j}|^2=E-\frac{|z|^2}{N}}f \circ R^T d\gamma^{N(d-1)}_{\sqrt{E-\frac{|z|^2}{N}}}\]
\[=\frac{\left\lvert \mathbb{S}^{d(N-j-1)-1} \right\rvert}{\left\lvert \mathbb{S}^{d(N-1)-1} \right\rvert \left(E-\frac{|z|^2}{N} \right) ^{\frac{d(N-1)-2}{2}}} \]
\[\int_{\sum_{i=1}^{j-1} |\xi_i|^2 + |\widetilde{\xi_j}|^2 \leq E-\frac{z|^2}{N}}d\xi_1 \dots d\xi_{j-1}d\widetilde{\xi_j} \left(E-\frac{|z|^2}{N}-\sum_{i=1}^{j-1} |\xi_i|^2 -|\widetilde{\xi_j}|^2\right)^{\frac{d(N-j-1)-2}{2}}\]
\[\int_{\mathbb{S}^{d(N-j-1)}\left(\sqrt{E-\frac{|z|^2}{N}-\sum_{i=1}^{j-1} |\xi_i|^2 -|\widetilde{\xi_j}|^2} \right)}f \circ R^T d\gamma^{d(N-j)}_{\sqrt{E-\frac{|z|^2}{N}-\sum_{i=1}^{j-1} |\xi_i|^2 -|\widetilde{\xi_j}|^2}}.\]
Since
\[ E-\frac{|z|^2}{N}-\sum_{i=1}^{j-1} |\xi_i|^2 -|\widetilde{\xi_j}|^2
=E-\sum_{i=1}^j |\xi_i|^2 - \frac{|z-\sqrt{j}\xi_j|^2}{N-j}\]
\[=E-\sum_{i=1}^j |v_i|^2 - \frac{|z-\sum_{i=1}^j v_i|^2}{N-j},\]
we find that
\[\int_{\mathcal{S}_B^N(E,z)}f d\sigma^N_{E,z}=\frac{\left\lvert \mathbb{S}^{d(N-1)-j-1} \right\rvert}{\left\lvert \mathbb{S}^{d(N-1)-1} \right\rvert \left(E-\frac{|z|^2}{N} \right) ^{\frac{d(N-1)-2}{2}}}\cdot \left(\frac{N}{N-j}\right)^{\frac{d}{2}} \]
\[\int_{\sum_{i=1}^j |v_i|^2 + \frac{|z-\sum_{i=1}^j v_i|^2}{N-j} \leq E}dv_1 \dots dv_j \left(E-\sum_{i=1}^j |v_i|^2 - \frac{|z-\sum_{i=1}^j v_i|^2}{N-j} \right)^{\frac{d(N-j-1)-2}{2}}\]
\[\int_{\mathcal{S}_B^{N-j}\left(E-\sum_{i=i}^j|v_i|^2,z-\sum_{i=1}^j v_i\right)}f d\sigma^{N-j}_{E-\sum_{i=i}^j|v_i|^2,z-\sum_{i=1}^j v_i}.\]
\end{proof}

\begin{lemma}\label{lem: h_delta^n is in L^p }
The function $\widehat{h_\delta}^n$ defined in (\ref{eq: the fourier transform of h associated to the generating function}) belongs to $L^{q}\left(\mathbb{R}^{d+1} \right)$ for any $n>\frac{2(1+d)}{qd}$.
\end{lemma}
\begin{proof}
By the definition, it is sufficient to show that $\widehat{h_a}^j\widehat{h_b}^{n-j}$ is in $L^{q}\left(\mathbb{R}^{d+1} \right)$ for all $j=0,1,\dots,n$ and $a,b>0$ ($\widehat{h_a}$ was defined in the proof of Lemma \ref{lem: the fourier transform of the generating function}). Indeed
\[\int_{\mathbb{R}^{d+1}}\left\lvert \widehat{h_a}(p,t) \right\rvert^{jq}
\left\lvert \widehat{h_b}(p,t) \right\rvert^{(n-j)q}dpdt\]
\[=\int_{\mathbb{R}^{d+1}}\frac{e^{-|p|^2\left(\frac{2aqj\pi^2}{1+16\pi^2 a^2 t^2}+\frac{2bq(n-j)\pi^2}{1+16\pi^2 a^2 t^2} \right)}}{\left(1+16\pi^2 a^2 t^2 \right)^{\frac{dqj}{4}}\left(1+16\pi^2 b^2 t^2 \right)^{\frac{dq(n-j)}{4}}}dpdt\]
\[=C_d\int_{\mathbb{R}} \frac{\left( \frac{2aqj\pi^2}{1+16\pi^2 a^2 t^2}+\frac{2bq(n-j)\pi^2}{1+16\pi^2 b^2 t^2} \right)^{-\frac{d}{2}}}{\left(1+16\pi^2 a^2 t^2 \right)^{\frac{dqj}{4}}\left(1+16\pi^2 b^2 t^2 \right)^{\frac{dq(n-j)}{4}}}dpdt ,
\]
where $C_d=\int_{\mathbb{R}^d}e^{-|x|^2}dx$.\\
The behaviour at infinity is that of $t^{-d\left(\frac{nq}{2}-1 \right)}$ and thus we conclude that $\widehat{h_\delta}^n\in L^{q}\left(\mathbb{R}^d \times [0,\infty) \right)$ for any $n>\frac{2(1+d)}{qd}$.
\end{proof}

\begin{lemma}\label{lem: equivalence of a measure and a function}
Let $F(x)$ be a continuous function in $L^q\left( \mathbb{R}^{d+1} \right)$ for some $q>1$ and let $P$ be a probability measure such that for any $\varphi \in C_c \left( \mathbb{R}^{d+1} \right)$ we have 
\[ \int_{\mathbb{R}^{d+1}} \varphi(x)F(x)dx=\int_{\mathbb{R}^{d+1}} \varphi(x)dP(x). \]
Then $F\geq0$, $F(x)\in L^1 \left( \mathbb{R}^{d+1} \right)$ and $dP(x)=F(x)dx$.
\end{lemma}
\begin{proof}
Let $E$ be any bounded Borel set. Given an $\epsilon>0$ we can find open sets $U_1,U_2$ and compact sets $C_1,C_2$ such that $C_i\subset E \subset U_i$ for $i=1,2$, $P \left( U_1\backslash C_1 \right)<\epsilon$ and $\lambda\left( U_2\backslash C_2 \right)<\epsilon$ where $\lambda$ represents the Lebesgue measure. 
Defining $U=U_1\cap U_2$ and $C=C_1\cup C_2$ we find an open and compact sets, bounding $E$ between them, such that $P \left( U\backslash C \right)<\epsilon$ and $\lambda\left( U\backslash C \right)<\epsilon$.\\
By Uryson's lemma we can find a function $\varphi_{\epsilon}\in C_c \left( \mathbb{R}^{d+1} \right)$ such that $0\leq \varphi_{\epsilon}\leq 1$, $\varphi_{\epsilon}\vert_C=1$ and  $\varphi_{\epsilon}\vert_{U^c}=0$.\\
We have that
\[ \int_{\mathbb{R}^{d+1}}\left\lvert \chi_E-\varphi_{\epsilon}\right\rvert |F(x)|dx= 
\int_{U\backslash C}\left\lvert \chi_E-\varphi_{\epsilon}\right\rvert |F(x)|dx \]
\[ \leq \left(\int_{U\backslash C}dx\right)^{\frac{1}{q^\ast}}\cdot 
\left(\int_{U\backslash C}|F(x)|^{q}dx\right)^{\frac{1}{q}}
\leq \sqrt[q^\ast]{\epsilon}\cdot \left\lVert F \right\rVert_{L^q \left( \mathbb{R}^{d+1} \right)} ,\]
and
\[ \int_{\mathbb{R}^{d+1}}\left\lvert \chi_E-\varphi_{\epsilon}\right\rvert dP \leq \int_{U\backslash C}dP<\epsilon .\]
Since $\int_{\mathbb{R}^{d+1}}\varphi_{\epsilon} (x)F(x)dx =\int_{\mathbb{R}^{d+1}}\varphi_{\epsilon}dP$
we conclude that
\[ \left\lvert \int_E F(x)dx - P(E)\right\rvert 
\leq \int_{\mathbb{R}^{d+1}}\left\lvert \chi_E-\varphi_{\epsilon}\right\rvert |F(x)|dx
+\int_{\mathbb{R}^{d+1}}\left\lvert \chi_E-\varphi_{\epsilon}\right\rvert dP\]
\[ \leq \epsilon+\sqrt[q^\ast]{\epsilon}\cdot \left\lVert F \right\rVert_{L^q \left( \mathbb{R}^{d+1} \right)},\]
and since $\epsilon$ is arbitrary we find that for any bounded Borel set $E$, $P(E)=\int_E F(x)dx$.\\
Next, given any Borel set $E$, define $E_m=E\cap B_m(0)$. We have that $E_m\uparrow E$ and as such $P(E)=\lim_{m\rightarrow \infty}P(E_m)$. Using Fatu's lemma we find that\[\int_E F(x)dx=\int_{\mathbb{R}^d}\lim_{m\rightarrow \infty}\chi_{E_m}F(x)dx
\leq \liminf_{m\rightarrow \infty} \int_{E_m}F(x)dx=\liminf_{m\rightarrow \infty} P(E_m)=P(E).\]
If we'll prove that $F\in L^1 \left(\mathbb{R}^{d+1} \right)$ we would be able to use the Dominated Convergence Theorem to show equality in the above inequality and conclude that $dP=F(x)dx$. \\
Since $F$ is continuous, if $\textbf{Im}F(x_0)\not=0$ for one point, we can find a ball around it, $B_r(x_0)$ such that $\textbf{Im}F\not=0$ in the entire ball. Since any ball is a bounded Borel set we have that 
\[P\left( B_r (x_0)\right)=\int_{B_r(x_0)}F(x)dx\not\in \mathbb{R},\]
which is impossible. Thus $F$ is real valued. \\
A similar argument shows that $F$ is positive. Indeed, if $F(x_0)<0$ for one point we can find a ball around it, $B_r(x_0)$ such that $F<0$ in that ball. We have that 
\[0>\int_{B_r(x_0)}F(x)dx=P\left( B_r (x_0)\right),\]
again - impossible.\\ 
Thus $F\geq 0$ and we have that 
\[\int_{\mathbb{R}^{d+1}}|F(x)|dx=\int_{\mathbb{R}^{d+1}}F(x)dx\leq P\left(\mathbb{R}^{d+1}\right)=1,\]
completing our proof.
\end{proof}
The last two Lemmas provide the proof to Lemma \ref{lem: the convolution yields a function}.
\begin{proof}[Proof of Lemma \ref{lem: the convolution yields a function}]
Due to Lemma \ref{lem: h_delta^n is in L^p }, $\widehat{h_\delta}\in L^2\left(\mathbb{R}^{d+1} \right)\cap L^1 \left(\mathbb{R}^{d+1} \right) $ for all $n>\frac{2(1+d)}{d}$. As such, it has an inverse Fourier transform $F_n \in L^2\left(\mathbb{R}^{d+1} \right)\cap C \left(\mathbb{R}^{d+1} \right)$. Given any $\varphi \in C_c \left(\mathbb{R}^{d+1} \right)$ we have that 
\[\int_{\mathbb{R}^{d+1}}\varphi (u,v)dh^{\ast n}(u,v)=\int_{\mathbb{R}^{d+1}}\widehat{\varphi}(p,t)\overline{\widehat{h_\delta}^n}(p,t)dpdt=\int_{\mathbb{R}^{d+1}}\varphi(u,v)\overline{F_n}(u,v)dudv.\]
By Lemma \ref{lem: equivalence of a measure and a function} we conclude that $F_n \geq 0$ and that $dh(u,v)=F_n(u,v)dudv$.
\end{proof}

\end{document}